\documentclass[11pt]{article}

\usepackage{fullpage}
\usepackage{url}
\usepackage{xspace}
\usepackage{color}
\usepackage{xcolor}
\usepackage{graphics}
\usepackage[dvips]{epsfig}
\usepackage{comment}
\usepackage{amsmath}
\usepackage{amssymb}
\usepackage{amsfonts}
\usepackage{graphicx}
\usepackage{algorithm}
\usepackage{algpseudocode}
\usepackage[small,compact]{titlesec}
\usepackage{times,hyphenat}

\newtheorem{theorem}{Theorem}[section]
\newtheorem{lemma}{Lemma}[section]
\newtheorem{corollary}{Corollary}[section]
\newtheorem{claim}{Claim}[section]

\newtheorem{definition}{Definition}[section]

\newcommand{\qed}{\hfill $\Box$ \bigbreak}
\newenvironment{proof}[1][]{%
\ifthenelse{\equal{#1}{}}{\noindent {\bf Proof. }}
{\noindent {\bf Proof of #1}\\}%
}{\hfill\qed}
\newenvironment{proofclaim}[1][]{%
\ifthenelse{\equal{#1}{}}{\noindent {\bf Proof of the claim: }}
{\noindent {\bf Proof of #1. }}%
}{\hfill$\star$}
\algblockdefx{Rep}{EndRep}{\algorithmicrepeat{}}{}
\algblockdefx{Rept}{EndRep}[1]{\algorithmicrepeat{} #1 times}{}
\algrenewtext{EndRep}{\algorithmicend{} \algorithmicrepeat{}}

\makeatletter
\newcommand*{\bdiv}{%
  \nonscript\mskip-\medmuskip\mkern5mu%
  \mathbin{\operator@font div}\penalty900\mkern5mu%
  \nonscript\mskip-\medmuskip
}
\makeatother

\newcommand{\ie}{{\em i.e.,}\xspace}
\newcommand{\eg}{{\em e.g.,}\xspace}

\newcommand{\instance}{(r,x,y,\phi,\tau,v,t,\chi)}
\newcommand{\dist}[2]{dist({#1},{#2})}

\begin{document}

	\baselineskip  0.20in 
	\parskip     0.1in 
	\parindent   0.0in 

	\title{{\bf Almost Universal Anonymous Rendezvous in the Plane\footnote{This work was performed within Project ESTATE (Ref.
	ANR-16-CE25-0009-03), supported by French state funds managed by the ANR (Agence Nationale de la
	Recherche). Andrzej Pelc was partially supported by NSERC discovery grant 2018-03899
and by the Research Chair in Distributed Computing at the
Universit\'e du Qu\'ebec en Outaouais.}}}

	\date{}
	\newcommand{\inst}[1]{$^{#1}$}

	\author{
		S\'{e}bastien Bouchard\inst{1},
		Yoann Dieudonn\'e\inst{2},
		Andrzej Pelc\inst{3}\\
		Franck Petit\inst{1}\\
		\inst{1} Sorbonne Universit\'e, CNRS LIP6, INRIA, Paris, France\\
		E-mails: \url{{sebastien.bouchard, franck.petit}@lip6.fr}\\
		\inst{2} Laboratoire MIS \& Universit\'{e} de Picardie Jules Verne, Amiens, France.\\
		E-mail: \url{yoann.dieudonne@u-picardie.fr}\\
		\inst{3} Universit\'e du Qu\'ebec en Outaouais, Canada.\\
		E-mail: \url{Andrzej.Pelc@uqo.ca}\\
	}

	\maketitle

	\begin{abstract}
	Two mobile agents represented by points freely moving in the plane and starting at two different positions, have to meet. The meeting, called {\em rendezvous}, occurs when agents are at distance at most $r$ of each other and never move after this time, where $r$ is a positive real unknown to them, called the {\em visibility radius}. Agents are anonymous and execute the same deterministic algorithm. Each agent has a set of private {\em attributes}, some or all of which can differ between agents. These attributes are: the initial position of the agent, its system of coordinates (orientation and chirality), the rate of its clock, its speed when it moves, and the time of its wake-up. If all attributes (except the initial positions) are identical and agents start at distance larger than $r$ then they can never meet, as the distance between them can never change. However, differences between attributes make it sometimes possible to break the symmetry and accomplish rendezvous. Such instances of the rendezvous problem (formalized as lists of attributes), are called {\em feasible}.

	Our contribution is three-fold. We first give an exact characterization of feasible instances. Thus it is natural to ask whether there exists a single  algorithm that guarantees rendezvous for all these instances. We give a strong negative answer to this question: we show two sets $S_1$ and $S_2$ of feasible instances such that none of them admits a single rendezvous algorithm valid for all instances of the set. On the other hand, we construct a single algorithm that guarantees rendezvous for all feasible instances outside of sets $S_1$ and $S_2$. We observe that these exception sets $S_1$ and $S_2$ are geometrically very small, compared to the set of all feasible instances: they are included in low-dimension subspaces of the latter. Thus, our rendezvous algorithm handling all feasible instances other than these small sets of exceptions can be justly called {\em almost universal}.

	\vspace*{0.5cm}

\noindent
{\bf keywords:}  anonymous agent, rendezvous, symmetry breaking, plane

\vspace*{2cm}

	\end{abstract}

\section{Introduction} \label{sec:intro}

\subsection{The background and the problem}

Two mobile entities starting at different locations of some environment, have to meet. This task, extensively researched in the literature, is known as {\em rendezvous}. The rendezvous problem was considered both in the network environment, when entities model software agents navigating in a communication network, and in the geometric context, when entities model mobile robots circulating in some terrain. The aim of rendezvous may be exchange of information gathered about the environment or joint planning of some future task, such as network maintenance or decontamination of a terrain where human presence would be hazardous.

We consider a geometric version of rendezvous: mobile entities, called {\em agents}, are represented by points freely moving in the plane and starting at two different positions.
Rendezvous occurs when agents are at distance at most $r$ of each other (according to some absolute measure of length) and never move after this time, where $r$ is a positive real unknown to them, called the {\em visibility radius}.

Agents are anonymous and execute the same deterministic algorithm. Each agent has a set of private {\em attributes}, some or all of which can differ between agents. These attributes are: the position of the agent, its system of coordinates (orientation and chirality), the rate of its clock, its speed when it moves, and the time of its wake-up. If all attributes (except the initial positions) are identical and agents start at distance larger than $r$ then they can never meet, as the distance between them can never change. However, differences between attributes make it sometimes possible to break the symmetry and accomplish rendezvous. Such instances of the rendezvous problem (formalized as lists of attributes), are called {\em feasible}. Note that an instance is feasible, if there exists an algorithm, even  specifically designed for this instance given as input, that guarantees rendezvous for it. (However, agents executing this algorithm do not know which agent is which in the instance).

The central question considered in this paper is:

\begin{quotation}
Which instances of the rendezvous problem are feasible and how to guarantee rendezvous for as many of them as possible by a single algorithm?
\end{quotation}

\subsection{The model}

In order to formally define our model, we consider some absolute system $\Gamma$ of Cartesian coordinates, some absolute length unit normalized to 1, some absolute time unit normalized to 1, and some point in time called 0. The rest of the description is with respect to these absolute notions, unknown to the agents. The visibility radius of the agents is expressed in absolute length units.

Each agent has a private system of Cartesian coordinates with the origin at the starting point of the agent. The $x$-axis of this system is rotated with respect to the $x$-axis of the absolute system by an angle $0 \leq \phi<2\pi$, and the chirality of the private system with respect to $\Gamma$ is either +1, if after rotating the absolute system by angle $\phi$ both systems are the same up to a shift, or -1,  if after this rotation the $y$ axis of the absolute system has the opposite direction than that of the private system.

Each agent has a private clock, such that the lapse of time between consecutive tics of this clock, called the time unit of the agent, lasts $\tau$ units of the absolute time.
The clock of the agent starts at its wake-up which occurs at some absolute time $t \geq 0$. Each agent has some constant speed $v$ defined as the absolute distance it travels in an absolute time unit.
Whenever the agent moves, it does so with this constant speed $v$. Each agent defines its private unit of length as the distance it travels during its time unit. Thus, in absolute terms,
the length of a private unit of length of an agent is $\tau v$.

Agents execute the same deterministic algorithm.
There are two types of move instructions in this algorithm. The first type is go $(dir, d)$ which is executed by an agent as going $d$ units of length of the agent in direction $dir$ in its private system of coordinates.
In our algorithms we use directions $N,S,E,W$ as a shorthand. Directions $N$ and $S$ are, respectively, the positive and negative direction along the $y$-axis and directions $E$ and $W$ are, respectively, the positive and negative direction along the $x$-axis.
An example of such an instruction is go $(S,2)$ executed by going 2 units of length of the agent in the negative direction parallel to its $y$-axis.
The second type of instructions is wait $(z)$ which is executed by an agent as waiting idle for $z$ time units of the agent.
We assume that agents can measure angles and distances precisely, although in our algorithms these quantities are pretty simple: angles are rational multiples of $\pi$ and distances are rational numbers.

We call $A$ the agent woken up first and $B$ the other agent. In the case of simultaneous wakeup these names are given arbitrarily. For the sake of simplicity and without loss of generality, we consider all the attributes of agent $A$ to be the absolute ones: its system of coordinates is $\Gamma$, its unit of time is the absolute unit and its wake up time is 0. Its speed is normalized to 1. (Of course, agents do not know which of them is $A$ and which is $B$, this convention is only for description ease). Using this convention, an instance of the rendezvous problem can be given as a list of attributes of the agent $B$, together with the value of the visibility radius. More precisely, an instance is a tuple $\instance$, where $r>0$ is the visibility radius
(in the measure of length of agent $A$), $(x,y)$ are coordinates of the initial position of $B$ in the system of coordinates of $A$, $0\leq \phi <2\pi$ is the angle {such that the directions of $x$-axes of $A$ and $B$ are the same after rotating the system of $A$ by angle $\phi$ (in the counterclockwise direction of $A$)}, $\tau>0$ is the number of time units of $A$ between consecutive tics of the clock of $B$, $v>0$ is the speed of $B$ in time and length units of $A$, $t \geq 0$ is the time difference between the wakeup time of $B$ and wakeup time of $A$ in time units of $A$, and $\chi$ is +1 or -1 depending on whether the directions of $y$-axes of $A$ and $B$ agree or not after rotating the system of $A$ by angle $\phi$ {(in the counterclockwise direction of $A$)}.

Using the above formalization of an instance of the rendezvous problem we can reformulate our central question as:

\begin{quotation}
Which instances $\instance$ are feasible and how to guarantee rendezvous for as many of them as possible by a single algorithm?
\end{quotation}


	\subsection{Our results}

	Our contribution is three-fold. We first give an exact characterization of feasible instances. Thus it is natural to ask whether there exists a single  algorithm that guarantees rendezvous for all these instances. We give a strong negative answer to this question: we show two sets $S_1$ and $S_2$ of feasible instances such that none of them admits a single rendezvous algorithm valid for all instances of the set. On the other hand, we construct a single algorithm that guarantees rendezvous for all feasible instances outside of sets $S_1$ and $S_2$. We observe that these exception sets $S_1$ and $S_2$ are geometrically very small, compared to the set of all feasible instances: they are included in low-dimension subspaces of the latter. More precisely, while the set of all feasible instances contains
	a ball of positive radius in the space $\mathbb{R}^7$, the exception set $S_1$ is contained in a copy of the subspace $\mathbb{R}^3$ and the
	exception set $S_2$ is contained in a copy of the subspace $\mathbb{R}^4$.  Thus, our rendezvous algorithm handling all feasible instances other than these small sets of exceptions can be justly called {\em almost universal}.

	Our Algorithm {\tt AlmostUniversalRV} generalizes results both from \cite{CGKK} and from \cite{PY2} (in the latter case for two agents). Indeed, our algorithm guarantees rendezvous for all instances $\instance$ when wakeup of the agents is simultaneous (i.e., $t=0$), and
	either (1) $\tau=v=1$ is not satisfied or (2) orientations are different (i.e., $0<\phi <2\pi$) and the chirality is the same (i.e., $\chi=1$).
		This is exactly the set of instances for which the rendezvous algorithm from \cite{CGKK} works. On the other hand, our algorithm guarantees
		rendezvous for all instances $\instance$ for which $\tau=v=1$, $\phi=0$, $\chi=1$, and which satisfy the assumption $t>\dist{(0,0)}{(x,y)}-r$, where $dist$ is the Euclidean distance. For two agents, this is exactly the set of instances for which the gathering algorithm from \cite{PY2} works (this set of instances was called the set of good configurations in \cite{PY2}). Apart from these sets of instances, already handled by the algorithms from \cite{CGKK} and \cite{PY2}, respectively, our Algorithm {\tt AlmostUniversalRV} guarantees rendezvous for many more instances: all those outside the small exception sets. This was neither the case in \cite{CGKK} nor in \cite{PY2}.

There are two major differences between \cite{CGKK} and our study. The first is that we consider arbitrary time delays between wake-ups of agents, while the authors of \cite{CGKK} restricted attention to simultaneous wake-ups. The second difference is in the notion of rendezvous feasibility. The authors of \cite{CGKK} were interested in conditions under which rendezvous is feasible for \emph{all} possible initial positions of agents, formulated such conditions, and designed a single algorithm guaranteeing rendezvous for all initial positions if the conditions are satisfied. By contrast, we first look at each instance separately, and characterize those for which rendezvous is possible even using an algorithm dedicated to this instance. We call such instances feasible. Then we design an algorithm guaranteeing rendezvous for ``almost all'' feasible instances with a precise meaning of ``almost all''. It is interesting to notice that large classes of instances rejected by \cite{CGKK} are not only possible to solve by a dedicated algorithm but can in fact be handled all together by our single algorithm. Such are \eg many instances for which chiralities of the agents are different. As for \cite{PY2}, we differ from it  by
		considering a much larger context: possibly different coordinate systems, clock rates and speeds of the agents.

		\subsection{Related work}

The rendezvous problem for mobile agents, and more generally its version for many agents, called {\em gathering}, was extensively studied in the literature, in various environments.
Models under which it was investigated can be classified along two main dichotomies. The first of them concerns the way in which agents move: it can be either deterministic or randomized. The second dichotomy is according to the type of environment in which agents navigate: it can be a network modeled as a graph or a terrain modeled as the plane, possibly with obstacles. The present paper considers deterministic rendezvous in the plane.

An excellent survey of randomized rendezvous in various models is
\cite{alpern02b}, cf. also  \cite{alpern95a,alpern02a,anderson90}.
Deterministic gathering in networks was surveyed in \cite{Pe}.
Gathering many labeled agents in the presence of Byzantine agents was studied in \cite{BDD,DPP}.
The gathering problem was also studied in the context of oblivious robot systems in the plane, cf.
\cite{CP05,FPSW}, and fault-tolerant gathering of robots in the plane was studied, e.g., in \cite{AP06,CP08}.

Deterministic gathering in graphs of agents equipped with tokens used to mark nodes was considered, e.g., in~\cite{KKSS}. Deterministic rendezvous of two agents with unique labels was studied in \cite{DFKP,KM,TSZ}.
These papers considered the time of gathering in arbitrary graphs.
In \cite{DFKP} the authors showed a rendezvous algorithm polynomial in the size of the graph, in the length of the shorter
label and in the delay between the starting times of the agents. In \cite{KM,TSZ} rendezvous time was polynomial in the first two of these parameters and independent of the delay.
 In \cite{CKP,FP} the optimization criterion for rendezvous was the memory size of the agents:
it was considered in \cite{FP} for trees and in  \cite{CKP} for general graphs.
Memory needed for randomized gathering in the ring was discussed, e.g., in~\cite{KKPM08}.

Apart from the synchronous model when clocks of the agents tick at the same rate, measuring rounds, several authors considered asynchronous gathering in the plane \cite{BBDDP,CFPS,FPSW} and in networks
\cite{BCGIL,CLP,DGKKP,DPV,GP}. In \cite{CFPS,FPSW} agents were anonymous, but were assumed to have total or restricted capability of seeing other agents.  However,  in the latter scenario only connected initial configurations were discussed. In \cite{BBDDP}, a related task of approach of two agents at distance 1 was investigated and agents were assumed to have distinct integer labels, which were used to break symmetry.

Computational tasks in anonymous networks were studied in the literature, starting with the seminal paper \cite{A},
followed, e.g., by \cite{ASW, BV,KKV}. While the considered tasks, such as leader election in message passing networks or computing Boolean functions, differ
from rendezvous studied in the present paper, the main concern is usually symmetry breaking, similarly as in our case.

Deterministic rendezvous of anonymous agents in arbitrary anonymous graphs was previously studied in \cite{CKP,DP1,GP,PY}. Papers \cite{CKP,DP1} were concerned with the synchronous scenario.
The authors of \cite{GP} characterized initial positions that allow asynchronous rendezvous.
In \cite{PY}, the authors considered the problem  of synchronous rendezvous of two anonymous agents in arbitrary graphs. They used wake-up time to break symmetry between the agents, but the situation was very different from our present setting. While
in \cite{PY} the authors designed a universal algorithm that guarantees rendezvous for all instances for which a dedicated rendezvous algorithm exists, in our setting such a single rendezvous algorithm in the plane cannot exist.

The two papers closest to our present study are \cite{CGKK} and \cite{PY2}. Each of them considered a different particular scenario: in \cite{CGKK} it was assumed that wake-up times of agents were simultaneous, while in \cite{PY2} the authors considered gathering of many agents but only with the same coordinate systems, clock rates and speeds. Our present paper generalizes both of them (the latter for the case of two agents).
The relation between our results and  those from  \cite{CGKK} and \cite{PY2} was described above in the subsection ``Our results''.

\section{Terminology and preliminaries} \label{sec:prelim}
	While all our algorithms are executed by interpreting distances, directions and time segments according to the local attributes of each agent, in our analysis we will understand all these values in the absolute system of coordinates, absolute time and length units, i.e., according to our convention, those of agent $A$ (unless explicitly specified otherwise).
	In particular $\dist{(x,y)}{(x',y')}$ denotes the Euclidean distance between points $(x,y)$ and $(x',y')$ in the system of coordinates of agent $A$, measured in its unit of length.

Without loss of generality, we assume that $r<\dist{(0,0)}{(x,y)}$. Otherwise rendezvous is accomplished immediately, and thus all instances with
$r\geq \dist{(0,0)}{(x,y)}$ are trivially feasible.

	 Instances $\instance$ for which $\tau=v=1$, i.e., those in which clock rates of both agents are the same and speeds are the same, are called {\em synchronous}.

	 We will use two procedures from the literature.
In \cite{CGKK}, the authors described a procedure guaranteeing rendezvous for all instances $\instance$ for which wakeup of the agents is simultaneous (i.e., $t=0$), and
either (1) the instance is not synchronous, or (2) orientations are different (i.e., $0<\phi <2\pi$) and the chirality is the same (i.e., $\chi=1$).
	They did not define precisely what are the allowed moves of the agents, but an inspection of their algorithms shows that they use two types of moves: straight segments and circles. Our model does not allow circles but it is easy to see that the procedure from \cite{CGKK}  remains valid (i.e., guarantees rendezvous for the same instances) if each circle is replaced by the square inscribed in it with sides parallel to the axes of the system of coordinates of  the agent. We call {\tt CGKK}  the procedure from \cite{CGKK} modified in this way.
	\footnote{
	Note that, as opposed to \cite{CGKK}, we needed to say precisely which moves are allowed because
some of our results are negative and thus we need to prove that some actions are impossible. In \cite{CGKK} this could be left implicit, as they only had positive results.
We decided for the easiest option with segment moves but is is easy to see that all our results (also the negative ones) remain valid if circles
are allowed as they were in \cite{CGKK}.}

       {\tt Latecomers} is the Algorithm GATHER(2) from \cite{PY2}. It guarantees rendezvous for all synchronous instances $\instance$ for which the systems of coordinates of the agents are a shift of each other, and satisfying the assumption $t>\dist{(0,0)}{(x,y)}-r$, i.e., for instances where
       $\tau=v=1$, $\phi=0$, $\chi=1$ and $t>\dist{(0,0)}{(x,y)}-r$. In \cite{PY2}, the authors considered the problem of gathering $n\geq 2$ agents with the same clock rates and speeds, and with systems of coordinates being shifts of one another. They designed an algorithm GATHER ($n$) which guarantees gathering under some condition on instances which is equivalent to $t>\dist{(0,0)}{(x,y)}-r$, for $n=2$.

We will need the following notion of {\em canonical line}.

	\begin{definition}[Canonical Line]\label{def:cano}
		We define the {\em canonical line} of an instance $\cal{I}=$ $\instance$ as follows:
		\begin{enumerate}
			\item if $\phi = 0$ this is the line parallel to the $x$-axes of both agents and equidistant from the origins of their respective coordinate systems;
			\item otherwise ($\phi \neq 0$), this is the line parallel to the bisectrix of the angle between the $x$-axes of the agents and equidistant from the origins of their respective coordinate systems;
		\end{enumerate}
	\end{definition}

An example of an instance and of its canonical line is depicted in Figure~\ref{fig:f1}.

\begin{figure}[httb!]
	\begin{center}
	\includegraphics[width=0.6\textwidth]{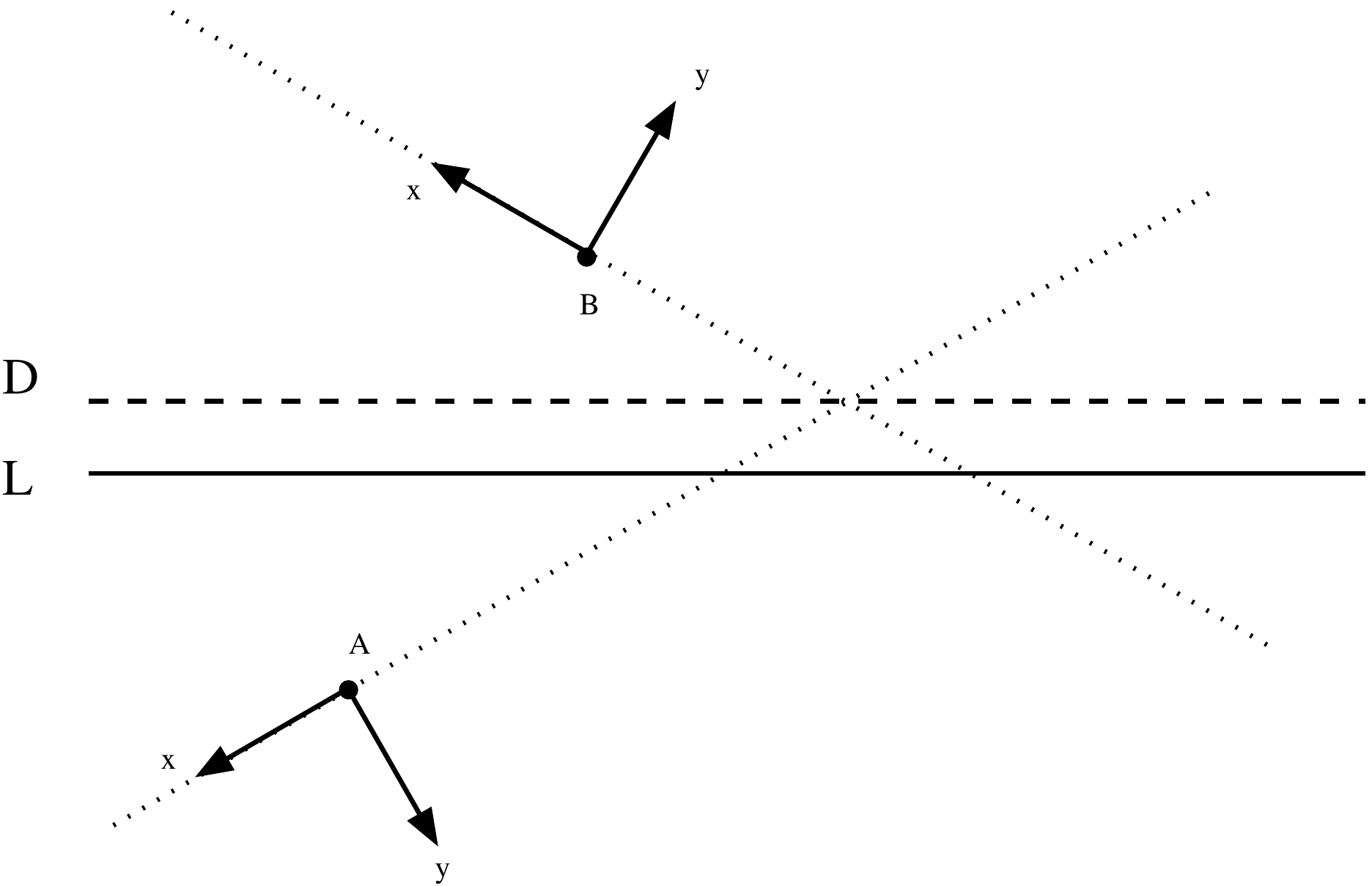}
	\caption{Illustration of the geometric setting of an instance with its canonical line in a case where the chiralities of the agents $A$ and $B$ are different. The dashed line $D$ represents the bisectrix of the angle between the $x$-axes of the agents, while the solid line $L$ represents the canonical line of the instance.}
	\label{fig:f1}
	\end{center}
\end{figure}

	We denote by $proj_A (s)$ (respectively, by $proj_B (s)$) the orthogonal projection of the position of agent $A$ (resp. agent $B$) at time $s$ on the canonical line of instance ${\cal I}$. To facilitate reading and when it is clear from the context, $proj_A$ (resp. $proj_B$) refers to $proj_A(0)$ (resp. $proj_B(0)$).

Sometimes, an agent will consider a local system {\tt Rot}$(\alpha)$: this system is the coordinate system resulting from rotating the system of the agent by the angle $\alpha$, counterclockwise w.r.t. this system.

We end this section by giving a simple lemma and its corollary that will be used recurrently throughout the paper.

	\begin{lemma}
		\label{lem:chi-dif}
		Consider any deterministic algorithm $\mathcal{A}$ and any synchronous instance ${\cal I}=\instance$ such that $\chi=-1$. Let  $z\geq t$ be a time such that rendezvous does not occur before time $z$ when applying algorithm $\mathcal{A}$ to the instance ${\cal I}$ and let $\vec{u}$ be the vector from $proj_A$ to $proj_B$. The trajectory followed by the later agent $B$ until time $z$ is an image of the trajectory followed by agent $A$ until time $z-t$ by a transformation that is a composition of a shift by vector $\vec{u}$ with the axial symmetry using the canonical line of~${\cal I}$.
	\end{lemma}

	\begin{proof}
		Agent $A$ starts time $t$ ahead of agent $B$ and instance ${\cal I}$ is synchronous: in particular $\tau=v=1$ and the length unit of each agent is $1$. Moreover, by Definition~\ref{def:cano}, the agents are initially at the same distance $d$ of the canonical line of ${\cal I}$, one on either side of this line if $d>0$. Hence, in view of $\chi=-1$ and since the agents apply the same deterministic algorithm, the lemma follows.
	\end{proof}

\begin{corollary}
		\label{cor:chi-dif}
		Consider any synchronous instance ${\cal I}=\instance$ such that $\chi=-1$ and any deterministic algorithm $\mathcal{A}$. Let $z\geq t$ be a time such that rendezvous does not occur before time $z$ when applying algorithm $\mathcal{A}$ to the instance ${\cal I}$. We have the following equality: $\dist{proj_A(z-t)}{proj_B(z)}=\dist{proj_A}{proj_B}$.
	\end{corollary}

\section{Characterization of Feasible Instances and Almost Universal Rendezvous} \label{sec:charac}

The aim of this section is to show two main results of this paper. The first result consists of a complete characterization of the feasible instances.
	\begin{theorem}
		\label{th:charac}
		~\begin{enumerate}
			\item All non-synchronous instances are feasible.
			\item A synchronous instance $\instance$ is feasible if and only if:
				\begin{enumerate}
					\item $\chi=1$ and $\phi\neq 0$ or
					\item $\chi=1$ and $\phi=0$ and $t\geq\dist{(0,0)}{(x,y)}-r$, or
					\item $\chi=-1$ and $t\geq\dist{proj_A}{proj_B}-r$.
				\end{enumerate}
		\end{enumerate}
	\end{theorem}

The second result is the design of an algorithm called {\tt AlmostUniversalRV}, detailed later in this section, which achieves rendezvous for the set of instances that is given by the following theorem: obviously all these instances are feasible.

	\begin{theorem}\label{th2}
		\label{th:algo}
		Algorithm {\tt AlmostUniversalRV} guarantees rendezvous for all instances $\instance$ which are either non-synchronous or such that:
		\begin{itemize}
			\item $\chi=1$ and $\phi\neq 0$ or
			\item $\chi=1$ and $\phi=0$ and $t>\dist{(0,0)}{(x,y)}-r$, or
			\item $\chi=-1$ and $t>\dist{proj_A}{proj_B}-r$.
		\end{itemize}
	\end{theorem}

Some feasible instances are not caught by our algorithm. Precisely, it is the case of the synchronous instances for which
\begin{itemize}
	\item either $\phi=0$, $t=\dist{(0,0)}{(x,y)}-r$, and $\chi=1$
	\item or $t=\dist{proj_A}{proj_B}-r$ and $\chi=-1$.
\end{itemize}

In Section~\ref{sec:miss}, we explain why these two sets of exceptions are small, in a geometric sense, compared to the set of all feasible instances, and we  show that no single algorithm can achieve rendezvous for all instances of either of these small sets.

Notice that the second theorem permits, roughly speaking, to show a big part of the ``if'' implication of the first theorem. Hence, in the rest of this section, we start by presenting Algorithm~{\tt AlmostUniversalRV} and proving Theorem~\ref{th:algo}, before proceeding with the proof of Theorem~\ref{th:charac} that will rely in part on Theorem~\ref{th:algo}.


\subsection{Algorithm {\tt AlmostUniversalRV} and proof of Theorem~\ref{th:algo}}


\subsubsection{Intuition} \label{subsec:intuition}

We begin with the following categorization of the instances described in Theorem~\ref{th:algo} into four disjoint types.

	\begin{itemize}
		\item{\bf type 1.} Synchronous instances for which $\chi=-1$ and $t>\dist{proj_A}{proj_B}-r$.
		\item{\bf type 2.} Synchronous instances for which $\chi=1$, $\phi=0$, and $t>\dist{(0,0)}{(x,y)}-r$.
		\item{\bf type 3.} Instances for which $\tau\neq 1$.
		\item{\bf type 4.} All instances described in Theorem~\ref{th:algo} that are not of type 1, 2, or 3.
	\end{itemize}

This categorization turns out to be pertinent for our purpose. Indeed, when thinking about the design of Algorithm {\tt AlmostUniversalRV}, it appears that each of the above four types of instances calls for its own rendezvous strategy.
Since combining these strategies into a single algorithm is pretty easy, and unnecessary to understand the intuitions, we focus below on describing the high-level idea of four algorithms, each of them being tailored to achieve rendezvous for all the instances of one of the above types.



First consider the instances of type $1$. For this type, there is a particular subset for which rendezvous can be easily achieved. This subset is the one composed of every instance ${\cal I}$  in which the agents are initially on the canonical line $L$ of ${\cal I}$, with their $x$-axes parallel to $L$ and going in the same direction. For instance ${\cal I}$, rendezvous can be achieved simply by instructing the agents to execute a linear search procedure, derived from the literature on the Cow-Path Problem \cite{BeNe}. This procedure consists in applying successive steps $i=1,2,\ldots$ till seeing the other agent, where in each step $i$ the sequence of move instructions is as follows: go East at distance $2^i$, then go West at distance $2^{i+1}$, and then go back to the initial position by going East at distance $2^i$. Using Corollary~\ref{cor:chi-dif} together with the fact that instance ${\cal I}$ is synchronous and $t>dist(proj_A,proj_B)-r$, it can be proved that the agents see each other by the time agent $A$ finishes step $\lambda$, where $\lambda$ is the smallest integer such that $t\leq2^\lambda$.

Actually, if the stop condition of the procedure were removed, we would even observe the agents getting regularly at a distance \emph{strictly} smaller than $r$ of each other from step $\lambda$ on. This is closely related to the fact that $t>dist(proj_A,proj_B)-r$.  If $t$ were exactly equal to $dist(proj_A,proj_B)-r$, the agents would get regularly at a distance $r$ but never less, while they would never see each other if $t$ were strictly smaller than $dist(proj_A,proj_B)-r$. Since getting the agents at distance \emph{exactly} $r$ of each other is already enough, there is room for possible deviations, i.e. a margin of error for which rendezvous could be still ensured even if the movements of the agents were not exactly on line $L$ and/or parallel to $L$. Obviously, this is true only if the magnitude of these deviations remains ``related'' to the magnitude of the margin of error. Although possibly very small, this margin turns out to be crucial to generalize the strategy described above to any instance of type $1$. This generalization can be done by requiring each agent to enumerate the triples of positive integers $(i,j,k)$, and for each triple $(i,j,k)$ to act as follows in a local system {\tt Rot}$(\frac{k\pi}{2^i})$: go North at distance $\frac{j}{2^i}$ and execute the first $i$ steps of the linear search procedure, then go South at distance $\frac{2j}{2^i}$ and execute again the first $i$ steps of the linear search procedure, and finally go back to the initial position by going North at distance $\frac{j}{2^i}$.

In doing so, we are guaranteed to recreate at some point almost the same favorable conditions as those we naturally have when the agents are initially on the canonical line and meet by the end of their execution of step $\lambda$. This can be proven using, among others, Lemma~\ref{lem:chi-dif} and the fact that $\chi=-1$. Note that under the aforementioned term ``almost the same favorable conditions'', we mean in particular that the agents get pretty close to the canonical line and start a linear search that will last enough time and that will not deviate much from the canonical line, with respect to the authorized margin of error, to achieve rendezvous.

To conclude with the case of instances of type $1$, it is interesting to observe that without room for error, it would be impossible to get a single algorithm achieving rendezvous for all the instances of type $1$. This is highlighted by the instances that are like those of type $1$ except that $t=dist(proj_A,proj_B)-r$ replaces $t>dist(proj_A,proj_B)-r$: we show in Section ~\ref{sec:miss} that there exists no algorithm achieving rendezvous for all these instances (cf. Theorem~\ref{th:impo}).

Concerning the instances of type 2, there is not much to say. This case can be directly solved via a procedure from the literature, called {\tt Latecomers}, introduced in the preliminaries section. Just note that the smooth functioning of procedure {\tt Latecomers} (resp. the approach described for the first type) heavily relies on the fact that $\chi=1$ and $\phi=0$ (resp. on the fact that $\chi=-1$), and thus it cannot be used to handle the instances of type~$1$ (resp. of type~$2$).

Now, let us turn attention to the instances of type $3$, which presents a major contrast with the previous two types. Indeed, we face here instances that are not synchronous, as each of them has the particularity of having an agent whose time unit is different from $1$. However, far from being a problem, this particularity can be exploited to our advantage. To see how, consider any planar search procedure $\mathcal{P}(a,b)$ that allows an agent, whose initial position is $p$ and whose unit of length is $u$, to get at distance at most $a u$ of all points at distance at most $b u$ from $p$. Although this question does not really matter here, note that there are multiple ways of designing such a procedure, for instance via spiral movements or via series of parallel linear searches: it is this latter option that will be used in our solution. Also consider an instance ${\cal J}$ of type $3$ and suppose ideally that the two agents of ${\cal J}$ initially know all parameters of the instance. The agents can thus compute the distance $d$ that initially separates them, {the smaller (resp. larger) of the two time units $\tau_{min}$ (resp. $\tau_{max}$) and the unit of length $u_X$ of the agent $X$ having time unit $\tau_{min}$}. Hence, they can also determine the smallest integer $\Delta$ such that $\Delta(\tau_{max}-\tau_{min})-t\geq n$, where $n$ is the number of local time units spent by an agent to execute {$\mathcal{P}(\frac{r}{u_{X}},\frac{d}{u_{X}})$} (this number is the same for all agents whatever the value of its unit of time). Note that integer $\Delta$ necessarily exists because $\tau_{max}>\tau_{min}$ in view of the definition of an instance of type~$3$.

In such an ideal situation, rendezvous can be achieved by requiring each agent to apply two steps that are stopped before the end as soon as an agent sees the other. Here are the two steps: $(1)$ wait time $\Delta$ (w.r.t my local time unit), and $(2)$ execute procedure {$\mathcal{P}(\frac{r}{u_{X}},\frac{d}{u_{X}})$}. Actually, integer $\Delta$ has been chosen so that {agent $X$} starts the second step at least time $n$ ahead of the other agent, even if agent $X$ is initially the later agent to wake up and to start the first step. As a result, agent $X$ can execute entirely procedure {$\mathcal{P}(\frac{r}{u_{X}},\frac{d}{u_{X}})$}, while the other agent stays idle at its initial position. {Since $u_X$ is the unit of length of agent $X$, procedure {$\mathcal{P}(\frac{r}{u_{X}},\frac{d}{u_{X}})$} allows agent $X$ to get at distance at most $r$ of all points at distance at most $d$ from its initial position.} Hence, agent $X$ necessarily sees the other agent and rendezvous is accomplished by the time it finishes executing the second step.

Unfortunately, we are never in such an ideal situation. That being said, with some appropriate arrangements, the same principle can be re-applied to achieve rendezvous from any instance of type~$3$. This can be done by instructing the agents to construct a series of assumptions about the useful parameters (i.e., $\tau_{min}$, $\tau_{max}$, $d$, etc.) and for each series to execute the aforementioned two steps using the newly made assumptions. For each series, we need to take the precaution of modulating the length of the waiting period of the second step according to the previous series that has been tested, in order to ensure that an agent stays idle while the other executes a planar search when the assumptions are good. In fact, it is true that we can never guarantee finding a series containing the exact values of the parameters. However, if the construction of the assumptions is conducted in an appropriate manner, we can obtain at some point lower and upper bounds on the parameters that are close enough to the reality to get the job done.

To finish with the intuitive explanations, it remains to address the case of the instances of type~$4$. These are all instances listed in Theorem~\ref{th:algo} that are of none of the previous three types. At first sight, they form a heterogeneous set. We find  there the non-synchronous instances for which $\tau=1$ and the synchronous instances for which $\chi=1$ and $\phi\ne0$ (since these latter instances are synchronous, they also have $\tau=1$). Each of these instances has at least one parameter whose value prevents it from being handled through one of the methods described earlier. In spite of their heterogeneity, the instances of type~$4$ all share two characteristics. The first characteristic is that they all have a time unit $\tau$ that is equal to $1$. The second characteristic, that relies on the function $h$ presented below, is that for each instance $\mathcal{K}$ of type~$4$, rendezvous can be achieved for instance $h(\mathcal{K})$ by applying a procedure from the literature, called {\tt CGKK}, introduced in the preliminaries section. The function $h$ associates to any instance another one that is identical except that the visibility radius of the agents is divided by $2$ and the delay between the starting times is set to $0$ (if it was not the case). The two above characteristics, which can be found together only in the instances of type~$4$, can be used to design a mechanism that is outlined below and that achieves rendezvous for any instance of this type.

For a better understanding of the idea behind this mechanism, suppose ideally that the agents of a given instance $\mathcal{K}$ of type~$4$ initially know the delay $t$ between their starting times as well as a positive number $\gamma$ such that no agent of $\mathcal{K}$ can travel a distance greater than $\frac{r}{4}$ during a period lasting time $\gamma$ or less. Under these ideal assumptions, rendezvous can be achieved through an algorithm that consists of an execution of procedure {\tt CGKK} that is interrupted after each period lasting a time $\gamma$, called a segment. Each of these interruptions corresponds to a waiting period lasting exactly time $t$, at the end of which the execution of {\tt CGKK} is resumed with the next segment. Naturally, as soon as an agent sees the other, it stops forever.

In order to explain why rendezvous is ensured here, we need to juggle two instances, namely $\mathcal{K}$ and its image $h(\mathcal{K})$, and thus we need to bring in some precisions to avoid any ambiguity. For instance $\mathcal{K}$, we keep the usual way to denote the reference agent by $A$ and the other agent by $B$, while for instance $h(\mathcal{K})$, the reference agent is denoted by $A'$ and the other agent by $B'$. We also keep the usual way of describing a situation with respect to the coordinate system and the parameters of a reference agent, using indistinctly those of agents $A$ or $A'$ as they are identical by definition.

The precisions being given, we can sketch our arguments. Consider an execution of procedure {\tt CGKK} for instance $h({\cal K})$. According to the properties of this procedure, agents $A'$ and $B'$ achieve rendezvous at some time $\alpha$ in which they occupy some positions $p'$ and $q'$ respectively. According to the construction of $h(\mathcal{K})$, we know that $p'$ and $q'$ are separated by a distance $\frac{r}{2}$. Now, consider the execution of the  algorithm described above for instance $\mathcal{K}$. When agent $A$ (resp. $B$)  executes the $\lceil\frac{\alpha}{\gamma}\rceil$th segment of procedure {\tt CGKK}, it can be shown that the agent passes through position $p'$ (resp. $q'$). Furthermore, each interruption of procedure {\tt CGKK} lasts exactly the time equal to the delay between the agents, and neither $A$ nor $B$ can travel a distance greater than $\frac{r}{4}$ during a segment. Consequently, if rendezvous has not occurred before, it can be also shown that when agent $A$ finishes its $\lceil\frac{\alpha}{\gamma}\rceil$th interruption, it is at distance at most $\frac{r}{4}$ from position $p'$, while at the same time agent $B$ starts its $\lceil\frac{\alpha}{\gamma}\rceil$th interruption and is at distance at most $\frac{r}{4}$ from position $q'$. Note that the above arguments rely on the fact that parameter $\tau$ is $1$ in instances ${\cal K}$ and $h({\cal K})$. From the fact that positions $p'$ and $q'$ are separated by a distance $\frac{r}{2}$, it follows that the agents of $\mathcal{K}$ get at distance at most $r$ of each other, and rendezvous is achieved by the time agent $A$ finishes its $\lceil\frac{\alpha}{\gamma}\rceil$th interruption of procedure {\tt CGKK}.

Of course, all of this works for an ideal situation in which the agents of instance $\mathcal{K}$ initially know the values of some relevant parameters. However, similarly to what has been described for the instances of type~$3$, the agents do not need the exact values of these parameters. Instead, they can guess and use good approximations that will turn out be sufficient to achieve rendezvous for all instances of type~$4$.

\subsubsection{Algorithm} \label{subsec:alg}

Algorithm~\ref{aurv:alg} gives the pseudocode of Algorithm {\tt AlmostUniversalRV}. It is mainly composed of a repeat loop that is interrupted as soon as the executing agent sees the other one.
This loop consists of four blocks of lines (marked by comments in Algorithm~\ref{aurv:alg}), each of them dedicated to the resolution of rendezvous for one of the types of instances mentioned earlier. Throughout the lines of each block, the reader will be able to recognize the corresponding strategy that is outlined in Section~\ref{subsec:intuition}.

{Algorithm {\tt AlmostUniversalRV} is designed to ensure rendezvous by the time agent $A$ finishes to execute the block $1\leq x\leq 4$ during the $i$th iteration of the repeat loop, if the instance is of type $x$ and satisfies some conditions depending on $i$ (e.g., the visibility radius is at least $\frac{1}{2^i}$, the delay between the agents is at most $2^i$, etc).
Actually, in the proof of correctness, we will show that for every instance ${\cal I}$ of one of the four types, there exists a specific value of $i$ such that rendezvous is guaranteed by the time any agent finishes the block dedicated to the type of ${\cal I}$ in the $i$th iteration of the repeat loop of Algorithm~\ref{aurv:alg}.


Algorithm~\ref{aurv:alg} relies on procedure {\tt PlanarCowWalk}, described in Algorithm~\ref{pcw:alg} which in turn relies on procedure {\tt LinearCowWalk} described in Algorithm~\ref{lcw:alg}. When an agent executes {\tt LinearCowWalk}$(i)$, for a positive integer $i$, it performs the first $i$ steps of a linear search on a line $l$ that is parallel to the $x$-axis of its local system. Each step $1\leq k \leq i$ permits the agents to visit all points of the line $l$ situated at distance at most $2^k$ (w.r.t its unit of length) from where it started {\tt LinearCowWalk}$(i)$.
When an agent executes {\tt PlanarCowWalk}$(i)$ from a point $p$, it performs {\tt LinearCowWalk}$(i)$ from each point whose coordinates expressed in its local system are as follows: the $x$-coordinate is the same as the one of point $p$ and the $y$-coordinate is $\frac{k}{2^i}$, where $k$ is an integer such that $|k|\leq 2^{2i}$. Procedure {\tt PlanarCowWalk} is used to achieve rendezvous both for the instances of type~$1$ and of type~$3$. In particular, for the instances of type~$3$, procedure {\tt PlanarCowWalk} plays the role of the planar search procedure mentioned in Section~\ref{subsec:intuition}: as we will see in the proof of Lemma~\ref{lem:type3}, the procedure allows an agent to get at distance at most $r$ of all points belonging to an area of a ``certain size" provided that the integer given as input is ``large enough''. For example, an execution of {\tt PlanarCowWalk}$(i)$ from a point $p$ ensures that the agent gets at distance at most $r$ of all points at distance at most $d$ from $p$, if its unit of length is at least $1$, $d\leq 2^i$ and $\frac{1}{2^i}\leq r$.

\begin{algorithm}
		\caption{{\tt AlmostUniversalRV}}
		\label{aurv:alg}
		\begin{algorithmic}[1]
                        \State execute the instructions of lines~\ref{aurv:line:1} to~\ref{aurv:line:last} given below and interrupt the execution as soon as the other agent is seen.\label{aurv:stop}
			\State $i\gets 1$ \label{aurv:line:1}
			\Rep
				\State /* instances of type~$1$ */
				\For{$j\gets 1$ to $2^{i+1}$}\label{aurv:debutphase}
					\State execute {\tt PlanarCowWalk}$(i)$ in the coordinate system {\tt Rot}$(\frac{j\pi}{2^i})$
				\EndFor \label{aurv:line:endfor}
				\State /* instances of type~$2$ */
				\State wait$(2^i)$ \label{aurv:line:8}
				\State execute {\tt Latecomers} during time $2^i$ \label{aurv:line:9}
				\State $P\gets$ the path followed in the latest execution of line~\ref{aurv:line:9} \label{aurv:line:9bis}
				\State backtrack on $P$ \label{aurv:line:10}
				\State /* instances of type~$3$ */
				\State wait$(2^{15i^2})$ \label{aurv:line:prevsearch}
				\State execute {\tt PlanarCowWalk}$(i)$ \label{aurv:line:search}
				\State /* instances of type~$4$ */
				\State let $S_1S_2\dots S_{2^{2i}}$ be the solo execution of {\tt CGKK} during time $2^i$, where each $S_j$ takes time $\frac{1}{2^i}$\label{aurv:line:11}
				\State execute $S_1${wait}$(2^i)$\dots $S_{2^{2i}-1}${wait}$(2^i)S_{2^{2i}}${wait}$(2^i)$ \label{aurv:line:12}
				\State $P'\gets$ the path followed in the latest execution of line~\ref{aurv:line:12} \label{aurv:line:12bis}
				\State backtrack on $P'$ \label{aurv:line:13}
				\State $i\gets i+1$
			\EndRep \label{aurv:line:last}
		\end{algorithmic}
	\end{algorithm}

	\begin{algorithm}
		\caption{{\tt PlanarCowWalk}$(i)$}
		\label{pcw:alg}
		\begin{algorithmic}[1]
			\State execute {\tt LinearCowWalk}$(i)$
			\For{$j\gets 1$ to $2$}
				\Rept{$2^{2i}$}
					\If{$j=1$}
						\State go $(N, \frac{1}{2^i})$
					\Else
						\State go $(S, \frac{1}{2^i})$
					\EndIf
					\State execute {\tt LinearCowWalk}$(i)$
				\EndRep
				\If{$j=1$}
					\State go $(S, 2^i)$
				\Else
					\State go $(N, 2^i)$
				\EndIf
			\EndFor
		\end{algorithmic}
	\end{algorithm}

\begin{algorithm}
		\caption{{\tt LinearCowWalk}$(i)$}
		\label{lcw:alg}
		\begin{algorithmic}[1]
			\For{$j\gets 1$ to $i$}
				\State go $(E, 2^j)$; go $(W, 2^{j+1})$; go $(E, 2^j)$
			\EndFor
		\end{algorithmic}
	\end{algorithm}

	\vspace*{4cm}

\subsubsection{Correctness}

In the sequel, we prove the correctness of Algorithm {\tt AlmostUniversalRV}. Overall, the proof is split into four lemmas, each of them showing that the algorithm achieves rendezvous for all instances of one of the four types.

In this proof, every execution by an agent of the repeat loop of Algorithm~\ref{aurv:alg} will be viewed as a series of consecutive \emph{phases} $i=1,2,3,\ldots$, where phase $i$ is the part of its execution corresponding to the $i$th iteration of the repeat loop. In a given phase $i$, every execution by an agent of the for loop of Algorithm~\ref{aurv:alg} will be viewed as a series of consecutive \emph{epochs} $j=1,2,3,\ldots,2^{i+1}$, where epoch $j$ is the part of its execution corresponding to the $j$th iteration of the for loop.

Before addressing the lemmas dedicated to the four types, we give the following simple lemma that will be used several times and that is related to the position occupied by an agent at some points during its execution of Algorithm {\tt AlmostUniversalRV}.

\begin{lemma}\label{lem:tech}
Consider Algorithm {\tt AlmostUniversalRV} executed for an instance ${\cal I}=\instance$. Each time agent $A$ (resp. $B$) starts the execution of a line of Algorithm~\ref{aurv:alg} whose number is not \ref{aurv:line:9bis}, \ref{aurv:line:10}, \ref{aurv:line:12bis} or \ref{aurv:line:13}, it does so from its initial position $(0,0)$ (resp. $(x,y)$).
\end{lemma}

\begin{proof}
In view of Algorithm~\ref{lcw:alg}, each execution by an agent of procedure {\tt LinearCowWalk} begins and ends at the same point. Thus, in view of Algorithm~\ref{pcw:alg}, each execution by an agent of procedure {\tt PlanarCowWalk} begins and ends at the same point as well. Moreover, during the execution of Algorithm {\tt AlmostUniversalRV}, each time an agent finishes processing line~\ref{aurv:line:10} (resp. line~\ref{aurv:line:13}) without having seen the other agent, it is at the point from which it started the previous processing of line~\ref{aurv:line:9} (resp. line~\ref{aurv:line:12}).

From the above explanations, it follows that within any given phase $i$ executed by an agent, each time it starts a line of Algorithm~\ref{aurv:alg} whose number is not \ref{aurv:line:9bis}, \ref{aurv:line:10}, \ref{aurv:line:12bis} or \ref{aurv:line:13}, it does so from the point at which it started phase $i$. It also follows that the execution of phase $i$ by an agent starts and ends at the same point if during this execution the agent does not see the other one.

As a result, we can prove by induction on $i$ that the lemma is true due to the fact that the execution by agent $A$ (resp. $B$) of phase $1$ starts from its initial position $(0,0)$ (resp. $(x,y)$).
\end{proof}


Now, we enter the core of the proof with the following lemma showing that Algorithm {\tt AlmostUniversalRV} can indeed deal with all instances of type~$1$.

	\begin{lemma}\label{lem:type1}
		Algorithm {\tt AlmostUniversalRV} guarantees rendezvous for all instances of type~$1$.
	\end{lemma}

	\begin{proof}
      Let ${\cal I}=\instance$ be an arbitrary instance of type $1$ and let $e=t-dist(proj_A,proj_B)+r$. Assume by contradiction that rendezvous does not occur by the time agent $B$ finishes executing the for loop of Algorithm~\ref{aurv:alg} during phase $i=\sigma+\omega$, where $\sigma$ and $\omega$ are defined as follows:


\begin{itemize}
\item $\sigma= \lceil \log (t + r + e + \sqrt{x^2+y^2} + \frac{8}{\min\{r,e\}} + \frac{\pi}{\arcsin{(\frac{\min\{r,e\}}{16(t+r+e+1)})}})\rceil$.
\item $\omega=\lceil\log(\frac{\pi}{\arccos{(\frac{dist(proj_A,proj_B)-r+\frac{e}{2}}{t})}})\rceil$ if $dist(proj_A,proj_B)-r+\frac{e}{2}>0$, $\omega=1$ otherwise.
\end{itemize}

In view of the definition of instance ${\cal I}$, we know that $t>dist(proj_A,proj_B)-r$. By the definition of $e$, we then have $e>0$ and {$dist(proj_A,proj_B)-r+\frac{e}{2}<t$}. Recall that $r>0$. Hence, we have $\frac{8}{\min\{r,e\}}>0$, $0<\frac{\min\{r,e\}}{16(t+r+e+1)}<1$ and $\arcsin{(\frac{\min\{r,e\}}{16(t+r+e+1)})}>0$. In the case where $dist(proj_A,proj_B)-r+\frac{e}{2}>0$, we also have $t>0$, {$0<\frac{dist(proj_A,proj_B)-r+\frac{e}{2}}{t}<1$} and $\arccos{(\frac{dist(proj_A,proj_B)-r+\frac{e}{2}}{t})}>0$. This means that $i$ is well-defined.

To proceed further, we need to introduce some notations and conventions.

Let $\Sigma$ be the system of coordinates resulting from the rotation of $\Gamma$ (the coordinate system of $A$) around its origin that ensures that the following two conditions are met: (1) the $x$-axis of $\Sigma$ is parallel to the canonical line $L$ of instance $\cal{I}$ and (2) $proj_A$ is East of $proj_B$ in $\Sigma$ or coincides with it. For ease of reading, the rest of this proof assumes the units of length and time of agent $A$ but assumes the coordinates system $\Sigma$, unless specified otherwise. Moreover, when in this proof we speak of an angle between two lines, it is always the smallest unoriented angle between them.


According to lines~\ref{aurv:debutphase} to~\ref{aurv:line:endfor} of Algorithm~\ref{aurv:alg}, there exists an integer $1\leq j\leq 2^{i+1}$ such that the system {\tt Rot}$(\frac{j\pi}{2^i})$ constructed by agent $A$ during epoch $j$ of phase $i$, which will be denoted by {\tt Rot}$_A(\frac{j\pi}{2^i})$, has the following two properties: $(1)$ its $x$-axis forms an angle $0\leq \alpha<\frac{\pi}{2^i}$ with the $x$-axis of $\Sigma$ (as well as with $L$), and $(2)$ the direction of its positive $x$-axis is between East (included) and North (excluded). Note that the North direction is naturally excluded in the second property because $i\geq2$ and thus $\alpha<\frac{\pi}{4}$. In the sequel, each time we refer to an epoch $j$, it is always the one of phase $i$, and thus we omit to mention it.

An example of the three coordinate systems $\Gamma$, $\Sigma$ and {\tt Rot}$_A(\frac{j\pi}{2^i})$ is depicted in Figure~\ref{fig:f2}.

\begin{figure}[httb!]
	\begin{center}
	\includegraphics[width=0.6\textwidth]{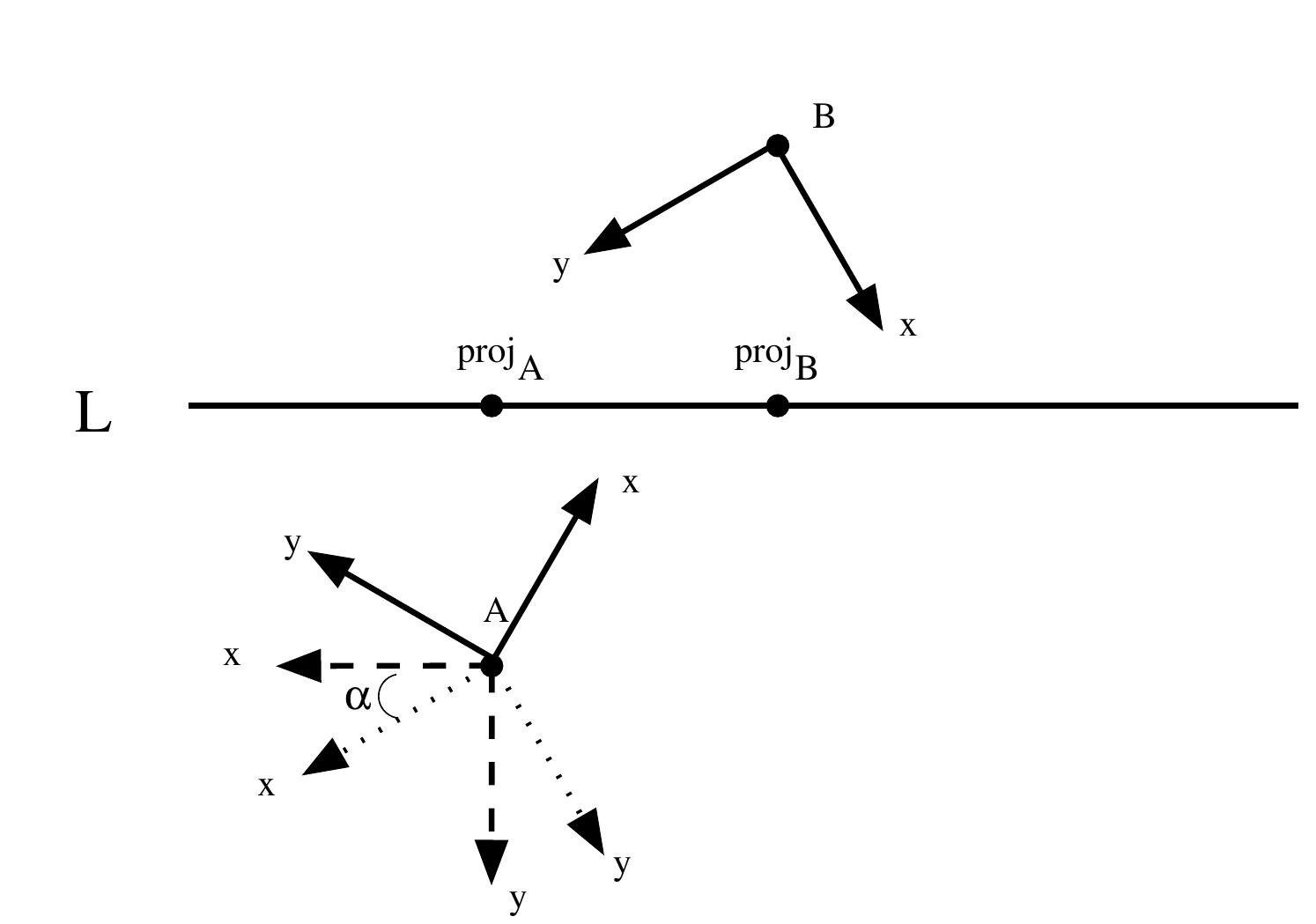}
	\caption{Illustration of the three coordinate systems $\Gamma$, $\Sigma$ and {\tt Rot}$_A(\frac{j\pi}{2^i})$ for an instance in which the chiralities of the agents $A$ and $B$ are different. The coordinate system $\Gamma$ is the one represented with the solid axes at the bottom of the figure, while $\Sigma$ (resp. {\tt Rot}$_A(\frac{j\pi}{2^i})$)  is the coordinate system represented with the dashed (resp. dotted) axes. As defined in the proof of Lemma~\ref{lem:type1}, the $x$-axis of $\Sigma$ (resp. {\tt Rot}$_A(\frac{j\pi}{2^i})$) forms an angle $0$ (resp. $\alpha$) with the canonical line $L$ of the instance.}
	\label{fig:f2}
	\end{center}
\end{figure}

Now we can start the substantial part of the proof with the following claim.


\begin{claim}\label{lem:type1:c1}
There is a time when agent $A$ starts executing procedure {\tt LinearCowWalk}$(i)$ during epoch $j$ from a point that is at a distance at most $\frac{\min\{r,e\}}{8}$ from line $L$.
\end{claim}

\begin{proofclaim}
In the proof of this claim, when we speak of some coordinates (resp. of the $x$-axis or $y$-axis), it is always implicit that they are expressed in (resp. they are those of) system {\tt Rot}$_A(\frac{j\pi}{2^i})$.

By Lemma~\ref{lem:tech}, agent $A$ is located at its initial position at the beginning of epoch $j$. By Definition~\ref{def:cano}, we know that the initial position of agent $A$ is at distance at most $\frac{\sqrt{x^2+y^2}}{2}$ from line $L$. By the definition of epoch $j$, we also know that $\alpha$ is the angle between $L$ and the $x$-axis. This means that the angle between the $y$-axis and the line perpendicular to $L$ passing through $proj_A$ is also $\alpha$ as depicted in Figure~\ref{fig:f3}. From these explanations, it follows that the distance between the initial position of agent $A$ and the intersection $o$ of the $y$-axis with $L$ is at most $\frac{\sqrt{x^2+y^2}}{2\cos(\alpha)}$. Since, in view of the definition of angle $\alpha$, we have $\cos(\frac{\pi}{4})<\cos(\alpha)\leq \cos(0)$, the distance between the initial position of agent $A$ and the intersection $o$ is at most $\sqrt{x^2+y^2}$, which is at most $2^i$ by the definition of $i$. In particular, the coordinates of point $o$ are $(0,y_o)$ where $y_o$ is some real belonging to $[-2^i,2^i]$.

By Algorithm~\ref{pcw:alg}, we know that for every integer $k\in\{-2^{2i},\ldots,2^{2i}\}$ agent $A$ executes procedure {\tt LinearCowWalk}$(i)$ from point $(0,\frac{k}{2^i})$ in epoch $j$. Moreover, by the definition of $i$, we have $\frac{1}{2^i}\leq\frac{\min\{r,e\}}{8}$. Hence, during epoch $j$, agent $A$ starts executing procedure {\tt LinearCowWalk}$(i)$ from a point $o'$ such that $dist(o,o')\leq\frac{\min\{r,e\}}{8}$, which concludes the proof of this claim.
\end{proofclaim}

\begin{figure}[httb!]
	\begin{center}
	\includegraphics[width=0.6\textwidth]{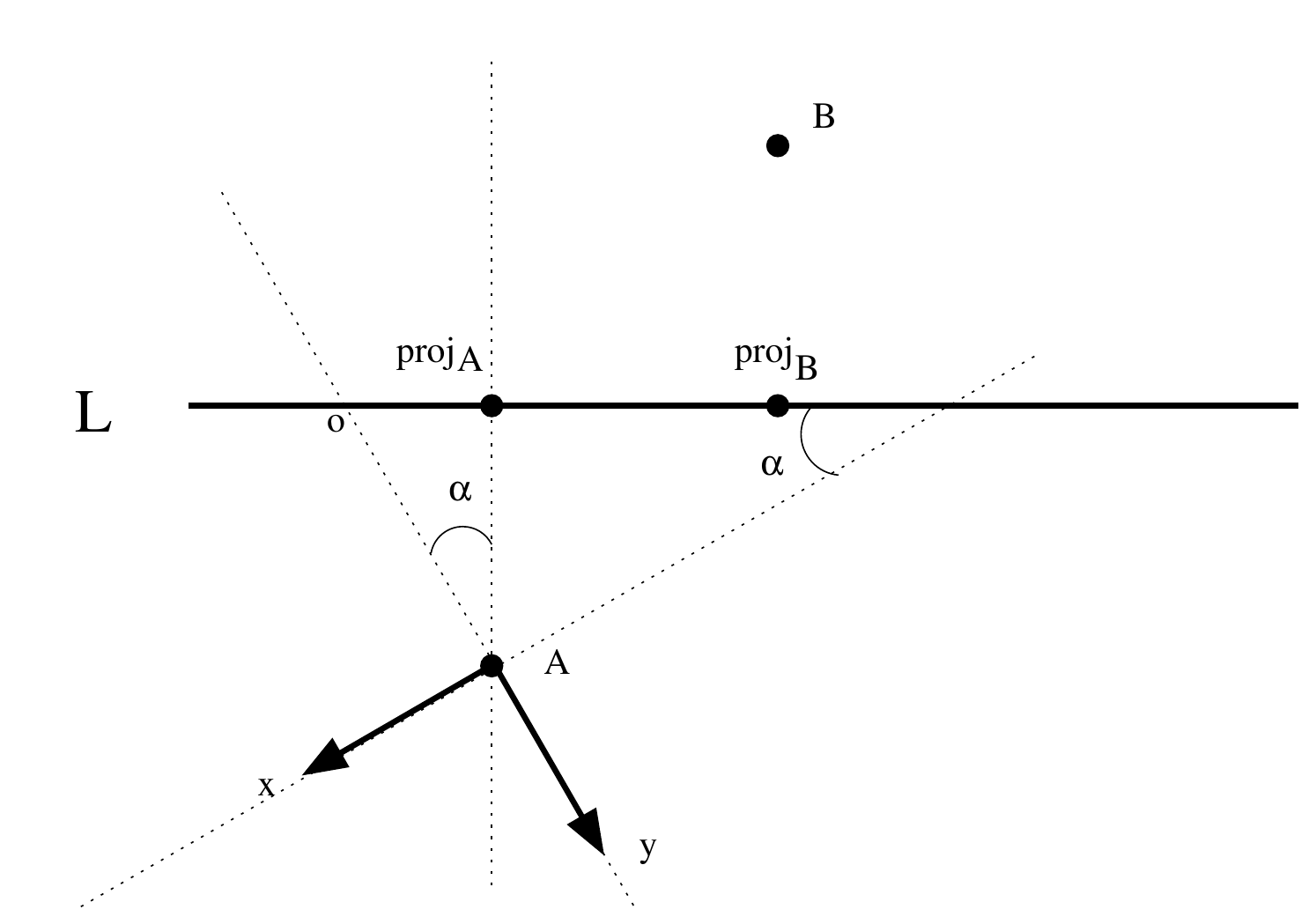}
	\caption{Illustration of some geometric arguments used in the proof of Claim~\ref{lem:type1:c1}}
	\label{fig:f3}
	\end{center}
\end{figure}

In the sequel, we denote by $s$ the first time when agent $A$ starts executing procedure {\tt LinearCowWalk}$(i)$ during epoch $j$ from a point that is at a distance at most $\frac{\min\{r,e\}}{8}$ from line $L$.

Let $L'$ be the line on which agent $A$ executes procedure {\tt LinearCowWalk}$(i)$ from time $s$ on. By the definition of epoch $j$, the angle between the $x$-axis of system {\tt Rot}$_A(\frac{j\pi}{2^i})$ and line $L$ is $\alpha$. Moreover, line $L'$ is parallel to the $x$-axis of {\tt Rot}$_A(\frac{j\pi}{2^i})$, in view of Algorithm~\ref{pcw:alg}. Hence we get the following claim:

\begin{claim}\label{lem:type1:c1bis}
The angle between line $L$ and line $L'$ is $\alpha$.
\end{claim}

\begin{claim}\label{lem:type1:c2}
Let $p$ and $q$ be two points separated by a distance $t$ and visited by agent $A$ during its execution of procedure {\tt LinearCowWalk}$(i)$ started at time $s$. The projections of $p$ and $q$ onto line $L$ are separated by a distance at least $dist(proj_A,proj_B)-r+\frac{e}{2}$.
\end{claim}

\begin{proofclaim}
If $dist(proj_A,proj_B)-r+\frac{e}{2}\leq0$ the lemma trivially holds. Hence, we assume in the proof of this claim that $dist(proj_A,proj_B)-r+\frac{e}{2}>0$.

Let $t_p$ (resp. $t_q$) be some time when agent $A$ occupies point $p$ (resp. $q$).

From Claim~\ref{lem:type1:c1bis}, we have the following equality

\begin{align}
\label{form:type1:1}
dist(proj_A(t_p),proj_A(t_q))= dist(p,q)\cos(\alpha)
\end{align}

By assumption, $dist(p,q)=t$. Moreover, by the definition of angle $\alpha$ and the fact that $i\geq 1$, we know that $0\leq\alpha<\frac{\pi}{2^i}\leq \frac{\pi}{2}$. Hence, from equality (\ref{form:type1:1}), we get

\begin{align}
\label{form:type1:2}
dist(proj_A(t_p),proj_A(t_q))\geq t \cos(\frac{\pi}{2^i})
\end{align}

By the definition of $i$, we also have the following inequality

\begin{align}
\label{form:type1:3}
i \geq \log(  \frac{\pi}   {      \arccos{(\frac{dist(proj_A,proj_B)-r+\frac{e}{2}}{t})}       }    )
\end{align}

The above inequality implies the one below

\begin{align}
\label{form:type1:4}
\frac{\pi}{2^i}\leq \arccos{(\frac{dist(proj_A,proj_B)-r+\frac{e}{2}}{t})}
\end{align}

From inequalities (\ref{form:type1:2}) and (\ref{form:type1:4}), it follows that $dist(proj_A(t_p),proj_A(t_q))\geq dist(proj_A,proj_B)-r+\frac{e}{2}$, which proves the claim.
\end{proofclaim}

Let $p_A$ be the position occupied by agent $A$ at time $s$. In view of Lemma~\ref{lem:tech} as well as Algorithms~\ref{aurv:alg} and~\ref{pcw:alg}, we know that the coordinates of $p_A$ are $(0,u)$ for some real $u$ in system {\tt Rot}$_A(\frac{j\pi}{2^i})$. Let $k$ be the smallest positive integer such that $t+r+e+1\leq2^k$ and let $\lambda$ be the time $s+\sum_{z=1}^{k-1}2^{z+2}$. In view of Algorithm~\ref{lcw:alg} and the fact that $i\geq k$, we know that in system {\tt Rot}$_A(\frac{j\pi}{2^i})$, agent $A$ moves East from point $p_A$ to a point $p'_A$ with coordinates $(2^k,u)$ during the time interval $[\lambda,\lambda+2^k]$. Then, during the time interval $[\lambda+2^k,\lambda+2^{k+1}]$, agent $A$ goes back to point $p_A$.

Note that, by definition, agent $A$ starts time $t$ ahead of agent $B$. Moreover, both agents execute Algorithm~\ref{aurv:alg} that is deterministic, and instance ${\cal I}$ is synchronous. Hence, we know that, in the system {\tt Rot}$(\frac{j\pi}{2^i})$ considered by agent $B$ in epoch $j$, agent $B$ moves East from a point $p_B$ with coordinates $(0,u)$ to a point $p'_B$ with coordinates $(2^k,u)$ during the time interval $[\lambda+t,\lambda+t+2^k]$. Then, during the time interval $[\lambda+t+2^k,\lambda+t+2^{k+1}]$, agent $B$ goes back to point $p_B$.

The move made by agent $A$ (resp. $B$) during the time interval $[\lambda,\lambda+2^k]$ (resp. $[\lambda+t,\lambda+t+2^k]$)  will be called its positive move, while the move made by agent $A$ (resp. $B$) during the time interval $[\lambda+2^k,\lambda+2^{k+1}]$ (resp. $[\lambda+t+2^k,\lambda+t+2^{k+1}]$) will be called its negative move.

\begin{claim}\label{lem:type1:c3}
During its positive and negative moves, agent $A$ as well as agent $B$ are always at distance at most $\frac{\min\{r,e\}}{4}$ from line $L$.
\end{claim}

\begin{proofclaim}
In the light of Lemma~\ref{lem:chi-dif} and the fact that the positive and negative moves of a given agent take place along the same segment, it is enough to show that the claim is true during the positive move of agent $A$.

According to the definitions of points $p_A$ and $p'_A$, these points are respectively the starting point and the endpoint point of the positive move of agent $A$ which occurs from time $\lambda$ to time $\lambda+2^k$. It follows that $proj_A(\lambda)$ (resp. $proj_A(\lambda+2^k)$) is the orthogonal projection of $p_A$ (resp. $p'_A$) on line $L$. By Claim~\ref{lem:type1:c1}, we have $dist(p_A,proj_A(\lambda))\leq\frac{\min\{r,e\}}{8}$. Moreover, by Claim~\ref{lem:type1:c1bis}, the angle between line $L$ and line $L'$, which passes through the segment $[p_A,p'_A]$, is $\alpha$. Hence, we can state that

\begin{align}
\label{form:type1:5}
dist(p'_A,proj_A(\lambda+2^k))\leq\frac{\min\{r,e\}}{8}+dist(p_A,p'_A)\sin(\alpha)
\end{align}

By the definition of a positive move, we know that $dist(p_A,p'_A)\leq 2^k$. By the definition of angle $\alpha$ and the fact that $i\geq1$, we also know that $0\leq\alpha<\frac{\pi}{2^i}\leq\frac{\pi}{2}$. From (\ref{form:type1:5}), we get the following inequality

\begin{align}
\label{form:type1:6}
dist(p'_A,proj_A(\lambda+2^k))\leq\frac{\min\{r,e\}}{8}+2^k\sin(\frac{\pi}{2^i})
\end{align}

By the definition of $i$, we also have the following inequality

\begin{align}
\label{form:type1:7}
i \geq \log(  \frac{\pi}   {     \arcsin{(\frac{\min\{r,e\}}{16(t+r+e+1)})}       }    )
\end{align}

This implies

\begin{align}
\label{form:type1:8}
\frac{\pi}{2^i}\leq  \arcsin{(\frac{\min\{r,e\}}{16(t+r+e+1)})}
\end{align}

By the definition of integer $k$, we know that $2^k\leq 2(t+r+e+1)$. As a result, from (\ref{form:type1:6}) and (\ref{form:type1:8}), we obtain

\begin{align}
\label{form:type1:9}
dist(p'_A,proj_A(\lambda+2^k))\leq\frac{\min\{r,e\}}{8}+2(t+r+e+1)\frac{\min\{r,e\}}{16(t+r+e+1)}\leq \frac{\min\{r,e\}}{4}
\end{align}

From Claim~\ref{lem:type1:c1} and formula (\ref{form:type1:9}), we know that each point of the segment $[p_A,p'_A]$ is at distance at most $\frac{\min\{r,e\}}{4}$ from line $L$, which concludes the proof of the claim.
\end{proofclaim}

\begin{claim}\label{lem:type1:c4}
The distance between $proj_A(\lambda+2^k)$ and $proj_B(\lambda+t+2^k)$ is $dist(proj_A,proj_B)$. Moreover, $proj_A(\lambda+2^k)$ is not West of $proj_B(\lambda+t+2^k)$.
\end{claim}

\begin{proofclaim}
By the definition of system $\Sigma$, we know that $proj_A$ is not West of $proj_B$. Thus, by Lemma~\ref{lem:chi-dif}, we know that $proj_A(\lambda+2^k)$ is not West of $proj_B(\lambda+t+2^k)$. Moreover, by Corollary~\ref{cor:chi-dif}, we have $dist(proj_A(\lambda+2^k),proj_B(\lambda+t+2^k))=dist(proj_A,proj_B)$, which proves the claim.
\end{proofclaim}

\begin{claim}\label{lem:type1:c5}
$proj_A(\lambda+2^k)$ (resp. $proj_B(\lambda+2^k)$) is not West (resp. not East) of $proj_A(\lambda+t+2^k)$ (resp. $proj_B(\lambda+t+2^k)$) and the distance between these two projections is at least $dist(proj_A,proj_B)-r+\frac{e}{2}$.
\end{claim}

\begin{proofclaim}
In view of Claim~\ref{lem:type1:c2} and the definition of the negative and positive moves, we know that the distance between $proj_A(\lambda-t+2^k)$ and $proj_A(\lambda+2^k)$ as well as the distance between $proj_A(\lambda+2^k)$ and $proj_A(\lambda+t+2^k)$ is at least $dist(proj_A,proj_B)-r+\frac{e}{2}$. By the definition of the negative and positive moves and the definition of system {\tt Rot}$_A(\frac{j\pi}{2^i})$, we know that  $proj_A(\lambda-t+2^k)$ is not East of $proj_A(\lambda+2^k)$ and $proj_A(\lambda+2^k)$ is not West of $proj_A(\lambda+t+2^k)$. From these explanations and Lemma~\ref{lem:chi-dif}, it also follows that $proj_B(\lambda+2^k)$ is not East of $proj_B(\lambda+t+2^k)$ and the distance between both these projections is at least $dist(proj_A,proj_B)-r+\frac{e}{2}$. This concludes the proof of the claim.
\end{proofclaim}

\begin{figure}[!htbp]
\begin{center}
  \begin{minipage}[t]{0.49\linewidth}
    \centering
	\includegraphics[width=1\textwidth]{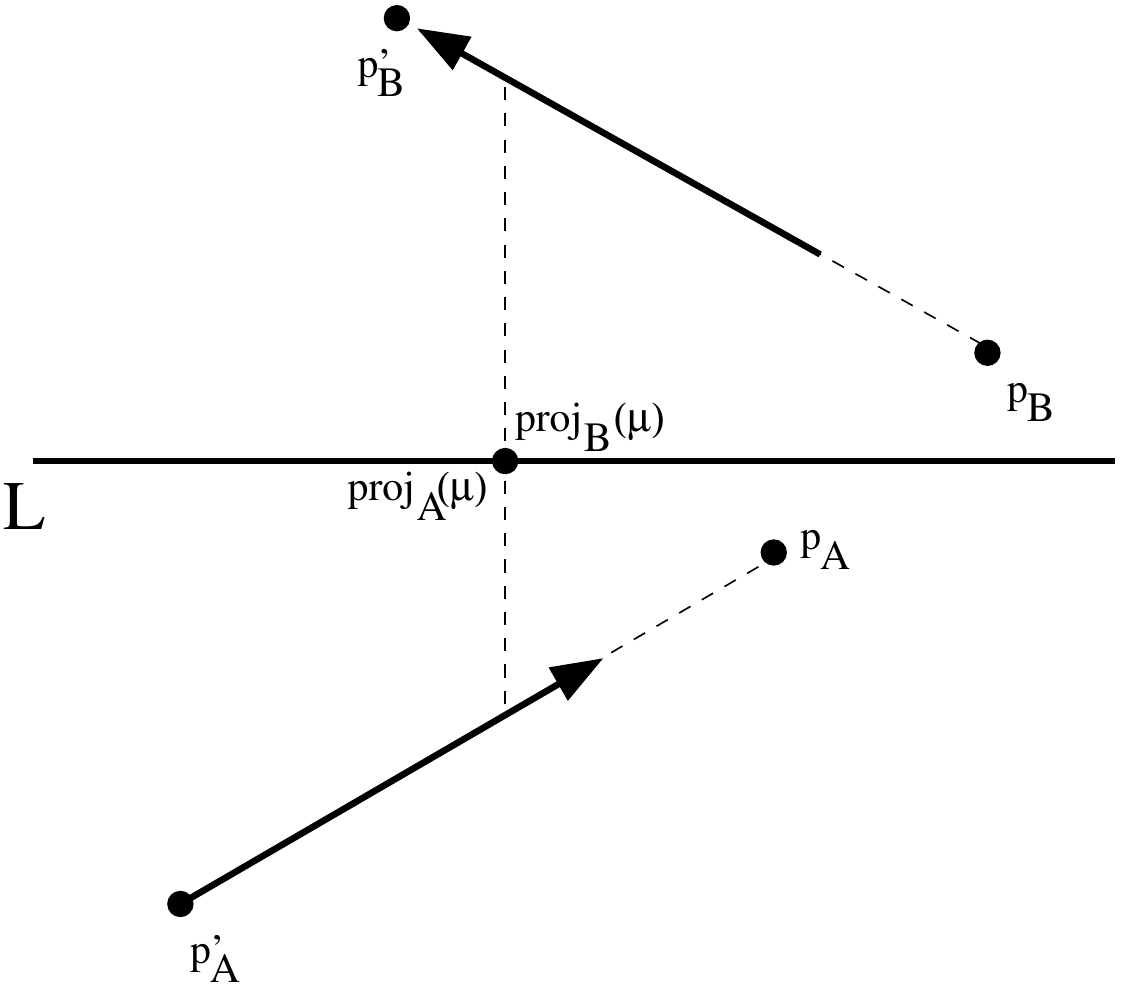}\\
    {\footnotesize ($a$)}
  \end{minipage}
  \begin{minipage}[t]{0.49\linewidth}
    \centering
	\includegraphics[width=1\textwidth]{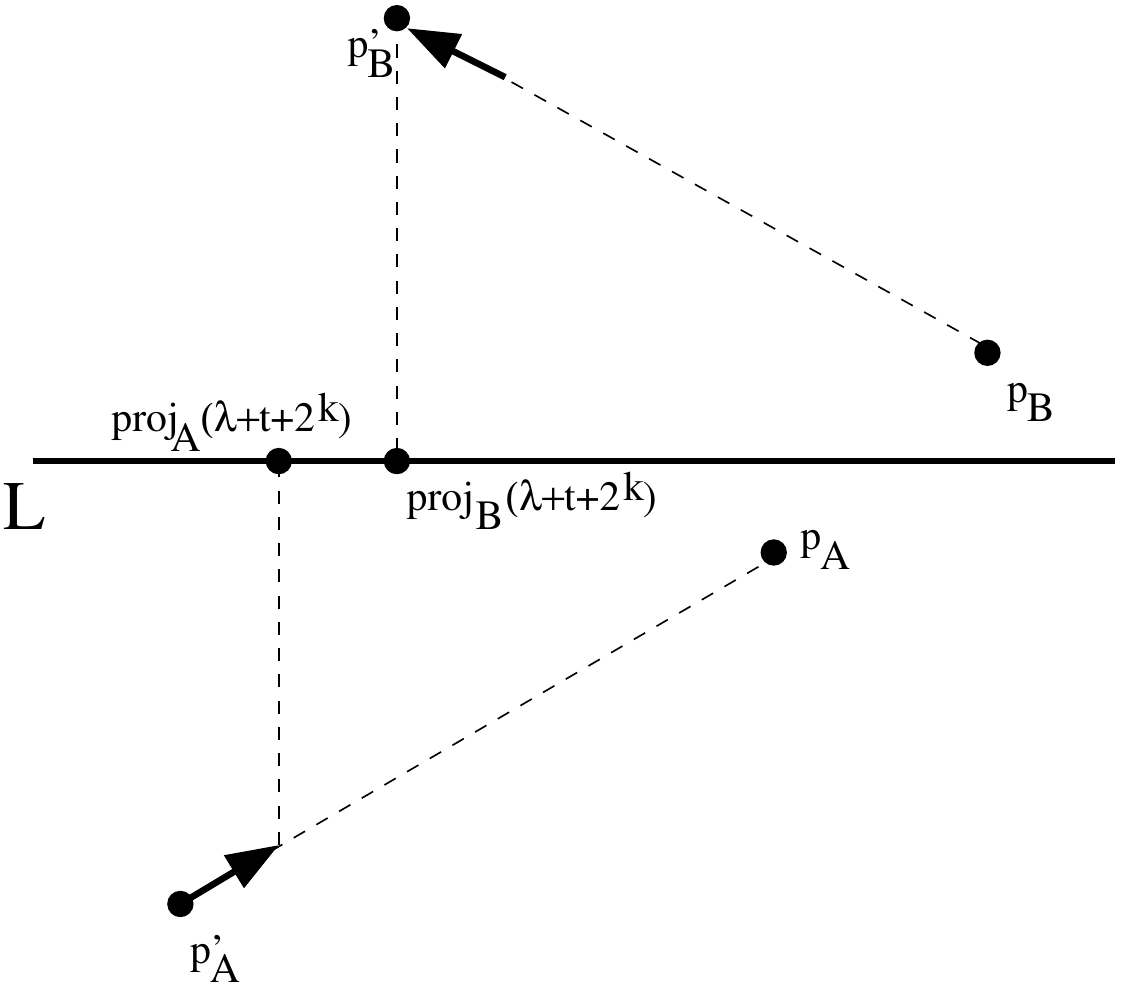}\\
    {\footnotesize ($b$)}
  \end{minipage}
\end{center}
 \caption{Examples of the two cases addressed at the end of the proof of Lemma~\ref{lem:type1}. The figure ($a$) corresponds to the case where the projections of $A$ and $B$ are the same at some point during the negative move of $A$, while the figure ($b$) corresponds to the complementary case. In each of these two figures, the arrow from $p'_A$ to $p_A$ (resp. $p_B$ to $p'_B$) represents the part of the negative move (resp. positive move) made by agent $A$ (resp. agent $B$) in the interval $[\lambda+2^k,\lambda+t+2^k]$.}
\label{fig4}
\end{figure}
Now, we want to argue that at some point during its negative move, agent $A$ is at a distance at most $r$ of agent $B$. There are two cases to consider.

First suppose that there is a time $\mu$ in the interval $[\lambda+2^k,\lambda+t+2^k]$ such that $proj_A(\mu)=proj_B(\mu)$. For example, this can occur when the projections of $A$ and $B$ cross each other at time $\mu$ as depicted in Figure~\ref{fig4}($a$). Since from time $\lambda+2^k$ to  time $\lambda+t+2^k$, agent $A$ (resp. $B$) executes a part of its negative (resp. positive) move, Claim~\ref{lem:type1:c3} implies that at time $\mu$ the agents are at distance at most $\frac{r}{4}$ of $proj_A(\mu)=proj_B(\mu)$ and thus at distance at most $\frac{r}{2}$ of each other.

{Now suppose that for every time $\mu$ in the interval $[\lambda+2^k,\lambda+t+2^k]$, we have $proj_A(\mu)\ne proj_B(\mu)$. For example, this can occur when the distance between the projections of $A$ and $B$ continuously decreases during the interval $[\lambda+2^k,\lambda+t+2^k]$ without reaching $0$, as depicted in Figure~\ref{fig4}($b$). Since the projections of $A$ and $B$ never coincide with each other during the interval $[\lambda+2^k,\lambda+t+2^k]$, Claims~\ref{lem:type1:c4} and~\ref{lem:type1:c5} imply that $dist(proj_A(\lambda+t+2^k),proj_B(\lambda+t+2^k))=dist(proj_A(\lambda+2^k),proj_B(\lambda+t+2^k))-\delta=dist(proj_A,proj_B)-\delta$, where $\delta$ is the distance traveled by the projection of $A$ during the interval $[\lambda+2^k,\lambda+t+2^k]$, which respects the double inequality $dist(proj_A,proj_B)-r+\frac{e}{2}\leq\delta<dist(proj_A,proj_B)$.}

{This means that $dist(proj_A(\lambda+t+2^k),proj_B(\lambda+t+2^k))$ is at most equal to $dist(proj_A,proj_B)-(dist(proj_A,proj_B)-r+\frac{\epsilon}{2})=r-\frac{\epsilon}{2}$. Note that, since for every time $\mu$ in the interval $[\lambda+2^k, \lambda+t+2^k]$, we have $proj_A(\mu)\neq proj_B(\mu)$, we can state that $dist(proj_A(\lambda+t+2^k),proj_B(\lambda+t+2^k))$ is positive and thus $r-\frac{\epsilon}{2}$ is positive. This implies  $r>\frac{\epsilon}{2}$.} Since from time $\lambda+2^k$ to  time $\lambda+t+2^k$, agent $A$ (resp. $B$) executes a part of its negative (resp. positive) move, Claim~\ref{lem:type1:c3} implies that at time $\lambda+t+2^k$ each agent is at distance at most $\frac{e}{4}$ from its projection onto line $L$. Hence, by Pythagoras' theorem, it follows that the distance between the agents at time $\lambda+t+2^k$, is at most $\sqrt{(\frac{e}{2})^2+(r-\frac{e}{2})^2}=\sqrt{\frac{e^2}{4}+r^2-2r\frac{e}{2}+\frac{e^2}{4}}$. Since $r>\frac{e}{2}$, we know that $\sqrt{\frac{e^2}{4}+r^2-2r\frac{e}{2}+\frac{e^2}{4}}\leq\sqrt{\frac{e^2}{4}+r^2-2\frac{e^2}{4}+\frac{e^2}{4}} \leq r$. Hence, the distance between the agents at time $\lambda+t+2^k$ is at most $r$.

Consequently, in all cases, agents $A$ and $B$ are at distance at most $r$, {and rendezvous is achieved in view of line~\ref{aurv:stop} of Algorithm~\ref{aurv:alg}, by the time agent $B$ completes the execution of the for loop of phase $i$}: this is a contradiction and proves the lemma.
\end{proof}

The following lemma shows that Algorithm {\tt AlmostUniversalRV} can handle every instance of type~$2$.

	\begin{lemma}\label{lem:type2}
		Algorithm {\tt AlmostUniversalRV} guarantees rendezvous for all instances of type $2$.
	\end{lemma}

	\begin{proof}
		Fix an arbitrary instance ${\cal I}=\instance$ of type $2$, and consider the execution $EX_1$ of procedure {\tt Latecomers} for ${\cal I}$. According to the properties of this procedure (cf. Section~\ref{sec:prelim}), $EX_1$ lasts a finite time $\Delta$, at the end of which rendezvous is achieved. Now, consider the execution $EX_2$ of Algorithm {\tt AlmostUniversalRV} for ${\cal I}$. Suppose by contradiction that rendezvous never occurs in $EX_2$ and denote by $s$ the time when agent $A$ starts executing line~\ref{aurv:line:9} of Algorithm~\ref{aurv:alg} during phase $\lceil \log (t + \Delta)\rceil$. Note that, during this phase, processing lines~\ref{aurv:line:8} and~\ref{aurv:line:9} of Algorithm~\ref{aurv:alg}  consists, in the local system of each agent, respectively in waiting a time at least $t$ and in executing procedure {\tt Latecomers} during a time at least $\Delta$. However, the units of time of both agents are identical because ${\cal I}$, which is an instance of type $2$, is by definition synchronous. Therefore, since rendezvous never occurs in $EX_2$, we know that agent $B$ waits during the time interval $[s,s+t]$ (in which it processes line~\ref{aurv:line:8} of Algorithm~\ref{aurv:alg}), starts executing procedure {\tt Latecomers} at time $s+t$ and keeps running it till at least time $s+t+\Delta$. We also know that agent $A$ starts executing procedure {\tt Latecomers} at time $s$ and keeps running it till at least time $s+\Delta$.
Moreover, Lemma~\ref{lem:tech} implies that the position occupied by agent $A$ (resp. $B$) in $EX_2$ at time $s$ is the same as the initial position of agent $A$ (resp. $B$) in $EX_1$.

From the above arguments, it follows that for each time $z\in[s,s+\Delta]$, the position occupied by agent $A$ (resp. agent $B$) in $EX_2$ is the same as the one occupied by agent $A$ (resp. agent $B$) at time $z-s$ in $EX_1$. In particular, the position occupied by agent $A$ (resp. agent $B$) at time $s+\Delta$ in $EX_2$ is the same as the one occupied by agent $A$ (resp. agent $B$) at time $s+\Delta-s=\Delta$ in $EX_1$. Note that, the positions occupied by the agents at time $s+\Delta$ in $EX_2$ are necessarily separated by a distance greater than $r$ as otherwise by Algorithm~\ref{aurv:alg} rendezvous occurs at this time during this execution. As a result, agents $A$ and $B$ are separated by a distance greater than $r$ at time $\Delta$ in $EX_1$. This implies, that rendezvous does not occurs at time $\Delta$ in $EX_1$, which is a contradiction and proves the lemma.
	\end{proof}


We continue with the next lemma that concerns the instances of type~$3$.

	\begin{lemma}\label{lem:type3}
		Algorithm {\tt AlmostUniversalRV} guarantees rendezvous for all instances of type~$3$.
	\end{lemma}

	\begin{proof}
Fix an arbitrary instance ${\cal I}=\instance$ of type $3$. By definition of such an instance, the two agents have different clocks. The agent with the faster clock will be called agent $X$, while the other agent will be called agent $Y$. Denote by $u_X$ (resp. $u_Y$) the unit of length of agent $X$ (resp. $Y$). Also denote by $\tau_X$ (resp. $\tau_Y$) the unit of time of agent $X$ (resp. $Y$). As usual, all times and durations (resp. distances) are expressed w.r.t the unit of time (resp. length) of the reference agent $A$ of ${\cal I}$, unless explicitely stated otherwise. Note that, agent $A$ is either $X$ or $Y$.

To proceed with the proof, we will show that rendezvous occurs by the time agent $X$ finishes the execution of line~\ref{aurv:line:search} of Algorithm~\ref{aurv:alg} in phase $i=\lceil\log (\frac{\tau_X}{\tau_Y-\tau_X} + \frac{\tau_Y}{\tau_X} + \frac{u_X}{r} + \frac{\sqrt{x^2+y^2}}{u_X}+t) \rceil$. Note that integer $i$ is well-defined as $\tau_Y>\tau_X>0$.

\begin{claim}\label{type3:c1}
			The execution by agent $X$ of procedure {\tt PlanarCowWalk}$(i)$ from its initial position $p_X$ allows it to get at distance at most $r$ of all points at distance at most $2^iu_X$ from $p_X$.
		\end{claim}
		\begin{proofclaim}
		In the proof of this claim all coordinates and distances are expressed in the system of $X$ (whose origin is $p_X$), except where otherwise specified.

Let $(f,g)$ be the cartesian coordinates of a given point located at distance at most $2^i$ from point $p_X$. We have $|f|\leq 2^i$, and therefore there exist two integers $-2^{2i}\leq a\leq 2^{2i}$ and $-\frac{1}{2^{i+1}}\leq b\leq \frac{1}{2^{i+1}}$ such that $f=\frac{a}{2^i}+b$. Similarly, there exist two integers $-2^{2i}\leq c\leq 2^{2i}$ and
$-\frac{1}{2^{i+1}}\leq d\leq \frac{1}{2^{i+1}}$ such that $g=\frac{c}{2^i}+d$. In view of Algorithm~\ref{pcw:alg}, we know that agent $X$ ends up executing {\tt LinearCowWalk}$(i)$ from point $(0,\frac{c}{2^i})$. By Algorithm~\ref{lcw:alg}, this execution of {\tt LinearCowWalk}$(i)$ allows the agent to traverse all the points $(f',\frac{c}{2^i})$ such that  $-2^i\leq f'\leq 2^i$. Hence, there is a time when the agent occupies the position $(\frac{a}{2^i},\frac{c}{2^i})$. The distance between $(f,g)$ and $(\frac{a}{2^i},\frac{c}{2^i})$ is $\sqrt{b^2+d^2}$, which is at most $\sqrt{2(\frac{1}{2^{i+1}})^2}<\frac{1}{2^i}$. Recall that this distance of $\frac{1}{2^i}$ is expressed in the unit of length of agent $X$. Expressed in the global unit, this distance is equal to $\frac{u_X}{2^i}$, which is at most $r$ by the definition of $i$. This concludes the proof of this claim.\end{proofclaim}

In phase $i$, the instruction at line~\ref{aurv:line:search} of Algorithm~\ref{aurv:alg} consists, for agent $X$, in executing {\tt PlanarCowWalk}$(i)$. By Lemma~\ref{lem:tech}, this execution starts from the initial position $p_X$ of agent $X$. Hence, in view of Claim~\ref{type3:c1} and the fact that, by the definition of $i$, $2^iu_X\geq \sqrt{x^2+y^2}$, we know that if agent $Y$ waits at its initial position $p_Y$ during the entire execution of line~\ref{aurv:line:search} of Algorithm~\ref{aurv:alg} in phase $i$ by agent $X$, rendezvous occurs by the time agent $X$ finishes this execution.

Note that if agent $Y$ executes the waiting period at line~\ref{aurv:line:prevsearch} of Algorithm~\ref{aurv:alg} during phase $i$, it is located at its initial position $p_Y$ according to Lemma~\ref{lem:tech}. So, in the light of the above arguments, it is enough to show the following two properties:
\begin{itemize}
\item If agent $X$ starts executing procedure {\tt PlanarCowWalk}$(i)$ at line~\ref{aurv:line:search} of Algorithm~\ref{aurv:alg} in phase $i$, agent $Y$ has at least started the waiting period of the previous line during the same phase.
\item Agent $Y$ cannot have started executing line~\ref{aurv:line:search} of Algorithm~\ref{aurv:alg} in phase $i$ before agent $X$ finishes the execution of the same line in the same phase.
\end{itemize}
This is done through Claims~\ref{type3:c2} to~\ref{type3:c4} that are proved below and that therefore conclude the proof of this lemma.

\begin{claim}\label{type3:c2}
The execution by agent $X$ (resp. $Y$) of the first $i-1$ phases of Algorithm~\ref{aurv:alg} and of the lines~\ref{aurv:debutphase} to~\ref{aurv:line:10} in phase $i$ of Algorithm~\ref{aurv:alg} lasts a time at most $2^{15(i-1)^2+4i+9}\tau_X$ (resp. $2^{15(i-1)^2+4i+9}\tau_Y$).
\end{claim}

\begin{proofclaim}

The execution of an instruction by agent $X$ takes time $c\tau_X$ for some real $c$ iff the execution of the same instruction by agent $Y$ takes time $c\tau_Y$. As a result, it is enough to show that the claim holds for agent $X$. In the rest of the proof of this claim, all durations are in the units of time of agent $X$.

Let $k$ be a positive integer. According to Algorithm~\ref{lcw:alg} (resp. Algorithm~\ref{pcw:alg}), the execution by agent $X$ of procedure {\tt LinearCowWalk}$(k)$ (resp. {\tt PlanarCowWalk}$(k)$) takes time at most $2^{k+3}$ (resp. at most $2^{3k+5}$). This implies that the execution by agent $X$ of lines~\ref{aurv:debutphase} to~\ref{aurv:line:9} (resp. lines~\ref{aurv:line:prevsearch} to~\ref{aurv:line:12}) of Algorithm~\ref{aurv:alg} in phase $k$ takes time at most $2^{4k+6}+2^{k+1}$ (resp. at most $2^{15k^2}+2^{3k+5}+2^{3k}+2^k$). The execution by agent $X$ of line~\ref{aurv:line:10} of Algorithm~\ref{aurv:alg} during phase $k$ is upper bounded by $2^k$. Moreover, since the execution by agent $X$ of line~\ref{aurv:line:12} of Algorithm~\ref{aurv:alg} in phase $k$ takes time at most $2^{3k}+2^k$ and consists in waiting at least time $2^{3k}$,  the execution of line~\ref{aurv:line:13} of Algorithm~\ref{aurv:alg} in phase $k$ takes time at most $2^k$.

Consequently, the execution time of phase $k$ by agent $X$ is upper bounded by $2^{4k+6}+2^{k+1}+2^{15k^2}+2^{3k+5}+2^{3k}+2^k+2^{k+1}$, which is at most $2^{15k^2+1}$.

Hence, the execution by agent $X$ of the first $i-1$ phases of Algorithm~\ref{aurv:alg} and of the lines~\ref{aurv:debutphase} to~\ref{aurv:line:10} in phase $i$ of Algorithm~\ref{aurv:alg} lasts a time at most $2^{4i+6}+2^{i+2}+\sum_{k=1}^{k=i-1}2^{15k^2+1}$. This is upper bounded by $2^{4i+7}+2^{15(i-1)^2+2}\leq 2^{15(i-1)^2+4i+9}$, which concludes the proof of this claim.
\end{proofclaim}

\begin{claim}\label{type3:c3}
If agent $X$ starts executing line~\ref{aurv:line:search} of Algorithm~\ref{aurv:alg} in phase $i$, agent $Y$ has at least started to execute line~\ref{aurv:line:prevsearch} of Algorithm~\ref{aurv:alg} in the same phase.
\end{claim}

\begin{proofclaim}
Let $q_X$ (resp. $q_Y$) be the time, if any, when agent $X$ (resp. $Y$) starts the execution of line~\ref{aurv:line:search} (resp. line~\ref{aurv:line:prevsearch}) of Algorithm~\ref{aurv:alg} in phase $i$.

$q_Y$ has higher value when agent $Y$ is agent $B$ than when it is agent $A$.
$q_X$ has lower value when agent $X$ is agent $A$ than when it is agent $B$.
In view of Claim~\ref{type3:c2} and the fact that $2^i\geq t$, by the definition of $i$, we have $q_Y\leq2^i+(2^{15(i-1)^2+4i+9})\tau_Y$ and
$q_X\geq(2^{15(i-1)^2+4i+9}+2^{15i^2})\tau_X$.


Consequently we have the following inequality

\begin{align}
\label{form:5}
q_X-q_Y\geq (2^{15(i-1)^2+4i+9}+2^{15i^2})\tau_X - 2^i-(2^{15(i-1)^2+4i+9})\tau_Y
\end{align}

By the definition of $i$, we have $2^i\geq\frac{\tau_Y}{\tau_X}$, which implies that $\tau_X\geq\frac{\tau_Y}{2^i}$. Hence, we get

\begin{align}
\label{form:6}
q_X-q_Y\geq (2^{15(i-1)^2+4i+9}+2^{15i^2})\frac{\tau_Y}{2^i} - 2^i-(2^{15(i-1)^2+4i+9})\tau_Y
\end{align}

The above inequality can be rewritten as follows

\begin{align}
\label{form:7}
q_X-q_Y\geq 2^{15i^2-i}\tau_Y -(2^{15(i-1)^2+4i+9})\frac{(2^i-1)\tau_Y}{2^i} - 2^i
\end{align}

This implies

\begin{align}
\label{form:7bis}
q_X-q_Y\geq 2^{15i^2-i}\tau_Y -(2^{15(i-1)^2+4i+9})\tau_Y - 2^i\geq (2^{15i^2-i}-2^{15(i-1)^2+4i+9})\tau_Y -2^i
\end{align}

Since, agent $Y$ is the agent with slower clock, we know that $\tau_Y\geq1$. Moreover, we know that $i\geq 1$ and thus $15(i-1)^2+4i+9\leq15i^2-i-1$. From (\ref{form:7bis}), we have

\begin{align}
\label{form:8}
q_X-q_Y\geq 2^{15i^2-i-1}-2^i>0
\end{align}

Consequently, the difference $q_X-q_Y$ is always positive, which concludes the proof of the claim.
\end{proofclaim}

\begin{claim}\label{type3:c4}
Agent $Y$ cannot have started executing line~\ref{aurv:line:search} of Algorithm~\ref{aurv:alg} in phase $i$ before agent $X$ finishes the execution of the same line in the same phase.
\end{claim}

\begin{proofclaim}
Let $s_X$ (resp. $s_Y$) be the time, if any, when agent $X$ (resp. $Y$) finishes (resp. starts) the execution of line~\ref{aurv:line:search} of Algorithm~\ref{aurv:alg} in phase $i$.

$s_Y$ has lower value when agent $Y$ is agent $A$ than when it is agent $B$.
$s_X$ has higher value when agent $X$ is agent $B$ than when it is agent $A$.
Note that according to Algorithms~\ref{lcw:alg} and~\ref{pcw:alg}, an execution by agent $X$ of {\tt PlanarCowWalk}$(i)$ takes time at most $2^{3i+5}\tau_X$. Moreover, by the definition of $i$, we know that $2^i\geq t$. Hence, in view of Claim~\ref{type3:c2}, we can state that $s_X\leq 2^i + (2^{15(i-1)^2+4i+9}+ 2^{15i^2} + 2^{3i+5})\tau_X$ and $s_Y\geq(2^{15(i-1)^2+4i+9}+2^{15i^2})\tau_Y$.

Consequently, we have the following inequality

\begin{align}
\label{form:1}
s_Y-s_X\geq(2^{15(i-1)^2+4i+9}+2^{15i^2})\tau_Y - 2^i - (2^{15(i-1)^2+4i+9} + 2^{15i^2} + 2^{3i+5})\tau_X
\end{align}

By the definition of $i$, we have $2^i\geq\frac{\tau_X}{\tau_Y-\tau_X}$, which implies that $\frac{2^i\tau_Y}{2^i+1}\geq\tau_X$. Hence (\ref{form:1}) implies

\begin{align}
\label{form:2}
s_Y-s_X\geq(2^{15(i-1)^2+4i+9}+2^{15i^2})\tau_Y - 2^i - (2^{15(i-1)^2+4i+9} + 2^{15i^2} + 2^{3i+5})\frac{2^i\tau_Y}{2^i+1}
\end{align}

Hence, we have

\begin{align}
\label{form:3}
s_Y-s_X\geq(2^{15(i-1)^2+4i+9}+2^{15i^2})\frac{\tau_Y}{2^i+1} - 2^i - 2^{3i+5}\frac{2^i\tau_Y}{2^i+1}
\end{align}

In view of the fact that $\tau_Y\geq 1$ and $i\geq1$, inequality ($\ref{form:3}$) implies that

\begin{align}
\label{form:4}
s_Y-s_X\geq\frac{2^{15i^2}\tau_Y}{2^i+1} - 2^i - 2^{3i+5}\frac{2^i\tau_Y}{2^i+1}\geq\frac{2^{15i^2}\tau_Y}{2^i+1}-\frac{2^{4i+6}\tau_Y}{2^i+1}>0
\end{align}

As a result, the difference $s_Y-s_X$ is always positive, which concludes the proof of this claim.
\end{proofclaim}


	\end{proof}

Now, it remains to deal with the instances of type~$4$. This is the aim of the following lemma.

	\begin{lemma}\label{lem:type4}
		Algorithm {\tt AlmostUniversalRV} guarantees rendezvous for all instances of type $4$.
	\end{lemma}

	\begin{proof}
Let ${\cal I}=\instance$ be an arbitrary instance of type $4$ and let ${\cal I}'=(\frac{r}{2},x,y,\phi,\tau,v,0,\chi)$ be an instance of type $4$ which is identical to ${\cal I}$ except that the visibility radius of the agents has been reduced by half and the delay between the starting times of the agents is $0$ (while it is not necessarily the case in ${\cal I}$). By the definition of type 4, instances ${\cal I}$ and ${\cal I}'$ belong to the set made of the non-synchronous instances for which $\tau=1$ and of the synchronous instances for which $\chi=1$ and $\phi\ne0$ (since these latter instances are synchronous, they also have $\tau=1$).

In instance ${\cal I}$, we keep our usual way to denote the reference agent by $A$ and the other agent by $B$. In instance ${\cal I}'$, the reference agent is denoted by $A'$ while the other agent is denoted by $B'$. We also keep the usual way to describe a situation with respect to the coordinate system and the parameters of a reference agent, using indistinctly those of agent $A$ or $A'$ as they are identical by convention. Note that, the four agents all have the same clock rate and agent $A$ (resp. agent $B$) shares the same unit of length with agent $A'$ (resp. $B'$). These elements must be kept in mind when reading the rest of this proof, as they condition the validity of  the arguments.

First consider the execution $EX'$ of procedure {\tt CGKK} for instance ${\cal I}'$. According to the properties of this procedure (cf. Section~\ref{sec:prelim}), rendezvous occurs in $EX'$ after a finite time $\Delta$, at the end of which the position $p'$ occupied by agent $A'$ is at distance at most $\frac{r}{2}$ from the position $q'$ occupied by agent $B'$.

Now, consider the execution $EX$ of Algorithm {\tt AlmostUniversalRV} for instance ${\cal I}$ and suppose by contradiction that rendezvous does not occur in $EX$ by the time agent $A$ finishes to execute phase $i=\lceil \log (t + \Delta + \frac{4(v+1)}{r})\rceil$. Let $S_{A,1}S_{A,2}\dots S_{A,2^{2i}}$ (resp. $S_{B,1}S_{B,2}\dots S_{B,2^{2i}}$) be the solo execution of {\tt CGKK} by agent $A$ (resp. $B$) during time $2^i$, where each segment $S_{A,j}$ (resp. $S_{B,j}$) takes time $\frac{1}{2^i}$.

Since $2^i\geq \Delta$, we know that agent $A$ (resp. $B$) executes at some point segment $S_{A,\lceil2^i\Delta \rceil}$ (resp. $S_{B,\lceil2^i\Delta \rceil}$) at line~\ref{aurv:line:12} of Algorithm~\ref{aurv:alg} during phase $i$. When it does so, we know that agent $A$ (resp. $B$) passes through position $p'$ (resp. $q'$) in view of execution $EX'$ and the fact that by Lemma~\ref{lem:tech} it always starts to execute line~\ref{aurv:line:12} of Algorithm~\ref{aurv:alg} from its initial position $(0,0)$ (resp. $(x,y)$). Since $2^i\geq t$ and the clock rates of the agents in instance ${\cal I}$ are identical, we also know that when agent $A$ finishes, say at some time $s$, the waiting period lasting $2^i$ just after the execution of segment $S_{A,\lceil2^i\Delta \rceil}$ in phase $i$, agent $B$ finishes or is still executing the waiting period lasting $2^i$ just after segment $S_{B,\lceil2^i\Delta \rceil}$ during the same phase.

As a result, at time $s$, agent $A$ (resp. $B$) is located at a position $p$ (resp. $q$) which is at distance at most $\frac{1}{2^i}$ (resp. $\frac{v}{2^i}$) from $p'$ (resp. $q'$) because the execution by agent $A$ (resp. $B$) of segment $S_{A,\lceil2^i\Delta \rceil}$ (resp. $S_{B,\lceil2^i(\Delta -t) \rceil}$) takes time $\frac{1}{2^i}$. By the definition of $i$, we know that $\frac{1}{2^i}\leq \frac{r}{4(v+1)}$. Hence, position $p$ (resp. $q$) is at distance at most $\frac{r}{4}$ from position $p'$ (resp $q'$). This implies that position $p'$ and $q'$ are separated by a distance greater than $\frac{r}{2}$ as otherwise, by Algorithm~\ref{aurv:alg} rendezvous occurs at time $s$ before agent $A$ finishes to execute phase $i=\lceil \log (t + \Delta + \frac{4(v+1)}{r})\rceil$. This contradicts the fact that the position $p'$ occupied by agent $A'$ is at distance at most $\frac{r}{2}$ from the position $q'$ occupied by agent $B'$ when rendezvous occurs in $EX'$. As a result, we know that rendezvous occurs in $EX$ by the time agent $A$ finishes executing phase $i=\lceil \log (t + \Delta + \frac{4(v+1)}{r})\rceil$, which concludes the proof of the lemma.
\end{proof}

Theorem \ref{th2} is a direct consequence of Lemmas \ref{lem:type1}, \ref{lem:type2}, \ref{lem:type3}, and \ref{lem:type4}.


\subsection{Proof of Theorem~\ref{th:charac}}\label{sec:proof31}

	We are now able to address the proof of Theorem~\ref{th:charac} that provides a characterization of the feasible instances.

	The next two lemmas follow from Theorem~\ref{th:algo}.
	\begin{lemma}
		\label{lem:1}
		Every non-synchronous instance is feasible.
	\end{lemma}

	\begin{lemma}
		\label{lem:2}
		Every synchronous instance $\instance$ for which $\chi=1 $ and $\phi\neq 0$ is feasible.
	\end{lemma}

	The following lemma follows from \cite{PY2}:
	\begin{lemma}
		\label{lem:3}
		Every synchronous instance $\instance$ for which $\chi=1$ and $\phi=0$ is feasible if and only if $t\geq\dist{(0,0)}{(x,y)}-r$.
	\end{lemma}

	The following lemma is the counterpart of the previous one for different chiralities. Notice that in this case, we do not need the assumption that $\phi=0$.

	\begin{lemma}
		\label{lem:4}
		Every synchronous instance $\instance$ for which $\chi=-1$ is feasible if and only if $t\geq\dist{proj_A}{proj_B}-r$.
	\end{lemma}

	\begin{proof}
		We first show the ``only if'' implication. By contradiction, consider a synchronous instance ${\cal I}$ for which $\chi=-1$ and $t<\dist{proj_A}{proj_B}-r$, and suppose that some deterministic algorithm $\mathcal{A}$ guarantees rendezvous for instance ${\cal I}$ at some time $z$. This implies $\dist{proj_A(z)}{proj_B(z)}\leq r$. Suppose that $z<t$. At time $z$, agent $B$ is still idle while agent $A$ has traversed a distance at most $z$. Hence, $\dist{proj_A(z)}{proj_B(z)}\geq\dist{proj_A}{proj_B} - z > \dist{proj_A}{proj_B} - t > r$ which is a contradiction. This proves $z-t\geq 0$, and in view of Corollary~\ref{cor:chi-dif}, we have $\dist{proj_A(z-t)}{proj_B(z)}=\dist{proj_A}{proj_B}$. Hence, $\dist{proj_A(z-t)}{proj_A(z)}\geq \dist{proj_A}{proj_B}-r$. Consequently, agent $A$ has to travel at least $\dist{proj_A}{proj_B}-r$ during time interval $[z-t, z]$ which contradicts the fact that $t<\dist{proj_A}{proj_B}-r$.

		We now show the ``if'' implication. By Theorem~\ref{th:algo}, we know that every synchronous instance for which $\chi=-1$ and $t>\dist{proj_A}{proj_B}-r$ is feasible. It is then enough to consider a synchronous instance ${\cal I}$ for which $\chi=-1$ and $t=\dist{proj_A}{proj_B}-r$ and to prove that ${\cal I}$ is feasible.

Let us describe a rendezvous algorithm working for ${\cal I}$. Each of the agents computes the canonical line $L$ of ${\cal I}$ which, by definition, has the same equation in the system of coordinates of each of them. Each agent considers its local coordinate system {\tt Rot}$(\frac{\phi+\pi}{2})$ and executes in this system the following instructions, which are interrupted as soon as it sees the other agent. The agent goes to the orthogonal projection on $L$ of its initial position. Then, it goes North at distance $t$ and then South at distance $t$.

		Recall that the canonical line $L$ of ${\cal I}$ is at the same distance from the initial positions of both agents. Also note that, for each agent, the last two moves are made along $L$, and direction North in the system {\tt Rot}$(\frac{\phi+\pi}{2})$ is the same for both agents.

		It remains to prove that the above algorithm guarantees rendezvous for instance ${\cal I}$.

Let $\Sigma$ be the system of coordinates {\tt Rot}$(\frac{\phi+\pi}{2})$ constructed by agent $A$ and let $z$ be the time when agent $A$ finishes its move North. At time $z$, agent $B$ has just reached line $L$. The rest of the arguments assumes the units of length and time of agent $A$  but the coordinates system $\Sigma$. {Note that we cannot have $proj_A=proj_B$ as this would imply $t=-r<0$, which would contradict the fact that $t\geq0$. Hence, it is enough to analyze the algorithm only when $proj_A\ne proj_B$. This is done considering the two cases below.}

\begin{figure}[!htbp]
\begin{center}
  \begin{minipage}[t]{0.40\linewidth}
    \centering
	\includegraphics[width=0.7\textwidth]{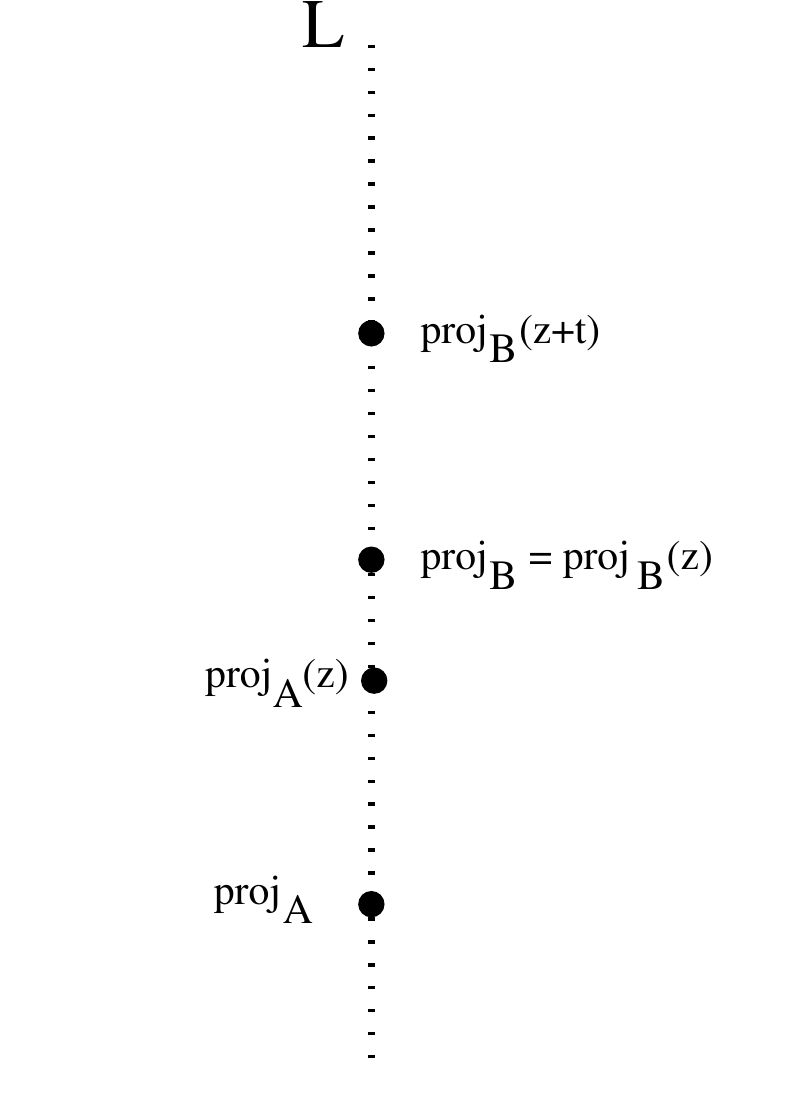}\\
    {\footnotesize ($a$)}
  \end{minipage}
  \begin{minipage}[t]{0.40\linewidth}
    \centering
	\includegraphics[width=0.7\textwidth]{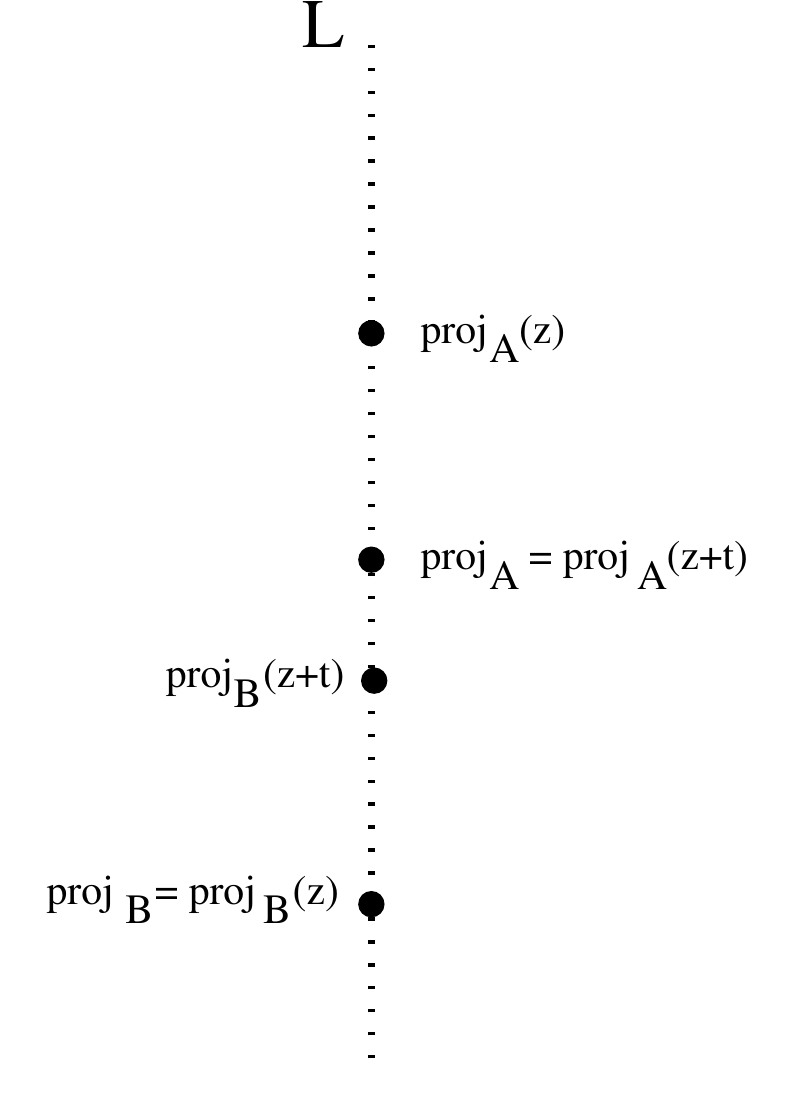}\\
    {\footnotesize ($b$)}
  \end{minipage}
\end{center}
 \caption{Examples of the two cases addressed in the proof of Lemma~\ref{lem:4}. The figures ($a$) and ($b$) correspond respectively to the cases $1$ and $2$. In each of the figures, we give the positions of the projections onto the canonical line $L$ at some relevant times. The North direction of $\Sigma$ is  bottom up in both figures.}
\label{fig5}
\end{figure}

		\begin{itemize}
			\item Case~1: $proj_B$ is North of $proj_A$.\\
			This case is depicted in Figure~\ref{fig5}($a$). In view of Corollary~\ref{cor:chi-dif}, we have $\dist{proj_A(z-t)}{proj_B(z)}=\dist{proj_A}{proj_B}$. In view of the description of the algorithm {and the fact that $t<\dist{proj_A}{proj_B}$}, agent $A$ moves towards $proj_B=proj_B(z)$ (without reaching it) during the time interval $[z-t, z]$. Hence, $\dist{proj_A(z)}{proj_B(z)}= \dist{proj_A}{proj_B}-t$, which is exactly $r$ because $t=\dist{proj_A}{proj_B}-r$. Thus, agents $A$ and $B$ are at distance $r$ of each other at time $z$ and, according to the algorithm, never move thereafter.
			\item {Case~2: $proj_B$ is South of $proj_A$}.\\
			This case is depicted in Figure~\ref{fig5}($b$). We have $\dist{proj_A(z+t)}{proj_B(z)}=\dist{proj_A}{proj_B}$ because $proj_A(z+t)$ (resp. $proj_B(z)$) is $proj_A$ (resp. $proj_B$). Moreover, agent $B$ moves towards $proj_A(z+t)$ (without reaching it) during the time interval $[z, z+t]$. Hence, $\dist{proj_A(z+t)}{proj_B(z+t)}=\dist{proj_A}{proj_B}-t$ which is, as in the first case, exactly equal to $r$. Hence, agents $A$ and $B$ are at distance $r$ of each other at time $z+t$ and, according to the algorithm, never move thereafter.
		\end{itemize}
		Hence, in all cases, rendezvous is achieved, which concludes the proof of the ``if'' implication, and thus the proof of the lemma.
	\end{proof}

Using the above lemmas, we can now prove Theorem~\ref{th:charac}.

	\begin{proof}[Theorem~\ref{th:charac}]
		The first statement of the theorem follows from Lemma~\ref{lem:1}. Hence, consider a synchronous instance. The ``if'' part of the second statement follows directly from Lemmas~\ref{lem:2},~\ref{lem:3}, and~\ref{lem:4}.

		It remains to prove the ``only if'' part of the second statement. First suppose that $\chi=-1$. In this case, the conjunction of the negations of 2a, 2b, and 2c implies $t<\dist{proj_A}{proj_B}-r$. Hence, the non-feasibility follows from Lemma~\ref{lem:4}. Next suppose that $\chi=1$. In view of the negation of 2a, we have $\phi=0$. Hence, the negation of 2b implies $t<\dist{(0,0)}{(x,y)}-r$, and the non-feasibility follows from Lemma~\ref{lem:3}.
	\end{proof}

\section{What Do We Miss} \label{sec:miss}

	It follows from Theorems~\ref{th:charac} and~\ref{th:algo} that the only feasible instances that are not handled by algorithm {\tt AlmostUniversalRV} are synchronous instances $\instance$ for which
	\begin{itemize}
	\item either $\phi=0$, $t=\dist{(0,0)}{(x,y)}-r$, and $\chi=1$
	\item or $t=\dist{proj_A}{proj_B}-r$ and $\chi=-1$.
	\end{itemize}
	Call these sets of instances $S_1$ and $S_2$ respectively.

	We now argue that these exception sets, while of course infinite, are small in a geometric sense, compared to the set of all feasible instances. First consider all feasible instances $\instance$ and
	partition them into two sets: $F_1$ are those with $\chi=1$ and $F_2$ are those with $\chi=-1$. Hence each of $F_1$ and $F_2$ can be formalized as a subset of $\mathbb{R}^7$. Since all non-synchronous instances are feasible, both for $\chi=1$ and for $\chi=-1$, and synchronous instances (i.e., those for which $\tau=v=1$) are contained in two copies of the subspace $\mathbb{R}^5$, it follows that each of  $F_1$ and $F_2$ contains a ball in $\mathbb{R}^7$
	of positive radius (intuitively these sets are ``fat'' in $\mathbb{R}^7$, i.e., not contained in any lower-dimension subspace of this space).

	Now consider the exception sets $S_1 \subset F_1$ and $S_2 \subset F_2$. $S_1$ is the set of synchronous instances in $F_1$ for which $\phi=0$ and $t=\dist{(0,0)}{(x,y)}-r$. Hence these are instances satisfying four (independent) linear equations and consequently their set is contained in a copy of a subspace $\mathbb{R}^3$ of $\mathbb{R}^7$.
	$S_2$ is the set of synchronous instances in $F_2$ for which $t=\dist{proj_A}{proj_B}-r$. Hence these are instances satisfying three (independent) linear equations and consequently their set is contained in a copy of a subspace $\mathbb{R}^4$ of $\mathbb{R}^7$. This shows that the ``fat'' set of feasible instances contains
	two exception sets that are very ``slim'' and Algorithm {\tt AlmostUniversalRV} is a single rendezvous algorithm handling all feasible instances except those two sets.  Another way of arguing that the exception sets are small compared to the set $F_1 \cup F_2$ of all feasible  instances is that the latter set has positive 7-dimensional Lebesgue measure (in fact it is easy to see that this measure is not only positive but infinite) while each of the exception sets  $S_1$ and $S_2$ has 7-dimensional Lebesgue measure 0.

	Moreover, it follows from \cite{PY2} that there is no single algorithm guaranteeing rendezvous for all instances from set $S_1$, and {it follows from Theorem~\ref{th:impo} below} that there is no single determinictic algorithm handling all instances from set $S_2$. Hence we miss little and cannot avoid it altogether.

	\begin{theorem}
		\label{th:impo}
		There does not exist an algorithm guaranteeing rendezvous for every synchronous instance $\instance$ such that $\chi=-1$ and $t=\dist{proj_A}{proj_B}-r$.
	\end{theorem}

	\begin{proof}
		Assume by contradiction that there exists an algorithm $\mathcal{A}$ guaranteeing rendezvous for every synchro\-nous instance such that $\chi=-1$ and $t=\dist{proj_A}{proj_B}-r$.

		Consider an instance ${\cal I}=\instance$ as above such that $t>0$ and suppose that agents execute algorithm $\mathcal{A}$. From Definition~\ref{def:cano}, the inclination of the canonical line of ${\cal I}$ is $\frac{\phi}{2}$. By inclination, we mean the smallest positive angle $\alpha$ such that the $x$-axis of agent A is parallel to the canonical line after rotating the system of A by angle $\alpha$. We have the following claim.

		\begin{claim}
			\label{th:impo:claim}
			Before rendezvous, the earlier agent $A$ has traversed at some point a non-nul segment of inclination $\frac{\phi}{2}$.
		\end{claim}

		\begin{proofclaim}
			When rendezvous occurs at a time $z$, we necessarily have $\dist{proj_A(z)}{proj_B(z)}\leq r$.

			Suppose that $z<t$. At time $z$, agent $B$ is still idle while agent $A$ has traversed a distance at most $z$. Hence, $\dist{proj_A(z)}{proj_B(z)}\geq\dist{proj_A}{proj_B} - z > \dist{proj_A}{proj_B} - t = r$ which is a contradiction. This proves $z-t\geq 0$.

			So, in view of Corollary~\ref{cor:chi-dif}, we have $\dist{proj_A(z-t)}{proj_B(z)}=\dist{proj_A}{proj_B}$. Consequently, we know that $\dist{proj_A(z-t)}{proj_A(z)}\geq \dist{proj_A}{proj_B}-r$ because $\dist{proj_A(z)}{proj_B(z)}\leq r$. Moreover, $\dist{proj_A(z-t)}{proj_A(z)}\leq t=\dist{proj_A}{proj_B}-r$ because the distance traveled by agent $A$ during the time interval $[z-t, z]$ is at most $t$. Hence, $\dist{proj_A(z-t)}{proj_A(z)}=t$. This can occur only if agent $A$ traverses a segment of length $t$ parallel to the canonical line during the time interval $[z-t, z]$. Since, the inclination of the canonical line is $\frac{\phi}{2}$, the claim is proved.
		\end{proofclaim}

		Now consider the solo execution of algorithm $\mathcal{A}$ by an agent and denote by $P$ the polygonal line forming the trajectory of the agent in its system of coordinates. This line is composed of a possibly infinite but countable sequence of segments $S_1, S_2, \dots$. In view of Claim~\ref{th:impo:claim}, we know that for every angle $\phi$, one of these segments must have inclination $\frac{\phi}{2}$. However,  the number of possible angles is uncountable. This is a contradiction.
	\end{proof}

\section{Conclusion} \label{sec:conclu}

	In our considerations, we assumed the same visibility radius $r$ for both agents, similarly as in \cite{CGKK,PY2}, in order to facilitate the reading. However, all our results remain true if the visibility radii are different. Assume that the visibility radius $r_1$ of one of the agents is not smaller than the visibility radius $r_2$ of the other one. Rendezvous is defined similarly as before: agents have to see each other (\ie be at distance at most $r_2$) and never move after this time.

	First note that all our negative results remain true with $r$ replaced by $r_1$: in all these results, agents will still never see each other. We now argue that all our positive results (\ie results concerning the algorithms) also remain correct, after replacing $r$ by $r_1$ in their validity conditions.

	Consider either Algorithm {\tt AlmostUniversalRV} or any algorithm working for a particular instance, executed as if both agents had the same radius $r_1$. We have the guarantee that at some point, both agents get at distance $r_1$. At this time, the agent with visibility radius $r_1$ sees the other agent and stops. Now the aim is for the other agent to get at distance $r_2$ from the agent with visibility radius $r_1$, in order to guarantee rendezvous. This will happen without any change in the case of algorithm {\tt AlmostUniversalRV}, as it contains in each phase a search procedure (\ie procedure {\tt PlanarCowWalk}), and can be done adding this search procedure as the last instruction in the other algorithms.

	A natural open problem is to generalize the rendezvous task to that of gathering many agents and to see how conditions on feasibility of instances change with respect to rendezvous of two agents. Which of the feasible instances can be gathered by a single algorithm? In a very restricted case, this problem has been solved in \cite{PY2}.

\bibliographystyle{plain}

\end{document}